\documentclass[runningheads,orivec,envcountsame]{llncs}

\usepackage[usenames,dvipsnames]{xcolor}
\usepackage{centernot}
\usepackage{amssymb}
\usepackage[linesnumbered,noend,boxruled]{algorithm2e}
\usepackage{stmaryrd}
\usepackage{bm}
\usepackage{tikz}
\usetikzlibrary{shapes,calc,arrows,automata,arrows.meta}
\usepackage{multirow}
\usepackage{array}
\usepackage{mathtools}
\usepackage{enumitem}
\usepackage[bookmarks,unicode,colorlinks=true]{hyperref}
\usepackage{proof}
\usepackage[capitalize,nameinlink]{cleveref}
\usepackage[location=appendix,manual]{moveproofs}
\usepackage{thm-restate}
\usepackage{subfig}
\usepackage{hhline}
\usepackage{graphicx}
\usepackage{microtype}
\usepackage{cite}
\usepackage{pgfplots}
\usepackage{listings}
\pgfplotsset{compat=1.18}
\usepackage[edges]{forest}
\usepackage{empheq}
\usepackage[font=small,skip=0pt]{caption}
\usepackage{wrapfig}
\usepackage{marvosym}

\everypar{\looseness=-1}
\allowdisplaybreaks[4]
\setlist{nosep}
\setlist{itemsep=1pt, topsep=1pt}
\AtBeginDocument{%
  \addtolength\abovedisplayskip{-0.3\baselineskip}%
  \addtolength\belowdisplayskip{-0.3\baselineskip}%
}

\SetKwIF{If}{ElseIf}{Else}{if}{do}{else if}{else}{end if}%

\SetCommentSty{mycommfont}

\colorlet{mygray}{gray!20}

\DontPrintSemicolon

\hypersetup{%
  pdftitle={Integrating Loop Acceleration into Bounded Model Checking},
  colorlinks=true,
  linkcolor=blue,
  citecolor=olive,
  filecolor=magenta,
  urlcolor=cyan
}

\makeatletter
\RequirePackage[bookmarks,unicode,colorlinks=true]{hyperref}%
   \def\@citecolor{blue}%
   \def\@urlcolor{blue}%
   \def\@linkcolor{blue}%

\def\orcidID#1{\href{http://orcid.org/#1}{\protect\raisebox{-1.25pt}{\protect\includegraphics{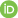}}}}
\makeatother

\spnewtheorem{requirement}{Requirement}{\bfseries}{\itshape}

\let\oldnl\nl
\newcommand{\nonl}{\renewcommand{\nl}{\let\nl\oldnl}}

\def\CPP{{\pl{C}\nolinebreak[4]\hspace{-.05em}\raisebox{.4ex}{\tiny\bf ++}}}

\renewcommand{\epsilon}{\varepsilon}
\let\oldphi\phi
\let\oldvarphi\varphi
\renewcommand{\varphi}{\oldphi}
\renewcommand{\phi}{\oldvarphi}

\newcommand{\pred}[1]{\mathsf{#1}}
\newcommand{\pl}[1]{\textsf{#1}}
\newcommand{\tool}[1]{\textsf{#1}}

\newcommand{\add}{\mathsf{add}}

\newcommand{\push}{\mathsf{push}}
\newcommand{\pop}{\mathsf{pop}}
\newcommand{\checksat}{\mathsf{check}\_\mathsf{sat}}
\newcommand{\getmodel}{\mathsf{get}\_\mathsf{model}}
\newcommand{\sip}{\mathsf{sip}}

\newcommand{\unsafe}{\mathsf{unsafe}}

\newcommand{\sat}{\mathsf{sat}}
\newcommand{\safe}{\mathsf{safe}}
\newcommand{\cache}{\mathsf{cache}}
\newcommand{\block}{\mathsf{b}}
\newcommand{\ABMC}{\mathsf{ABMC}}
\newcommand{\BMC}{\mathsf{BMC}}

\newcommand{\accel}{\mathsf{accel}}
\newcommand{\compose}{{\circledcirc}}
\newcommand{\concat}{\mathrel{::}}

\newcommand{\init}{\pred{init}}
\newcommand{\pre}{\mathsf{pre}}
\newcommand{\err}{\pred{err}}
\newcommand{\QF}{\mathsf{QF}}

\renewcommand{\AA}{\mathcal{A}}

\newcommand{\CC}{\mathcal{C}}
\newcommand{\LL}{\mathcal{L}}
\newcommand{\VV}{\mathcal{V}}

\newcommand{\trace}{\mathsf{trace}}
\newcommand{\inner}{\mathsf{i}}
\renewcommand{\outer}{\mathsf{o}}

\newcommand{\shouldaccel}{\mathsf{should\_accel}}

\newcommand{\ZZ}{\mathbb{Z}}

\newcommand{\NN}{\mathbb{N}}

\newcommand{\TT}{\mathcal{T}}

\newcommand{\DG}{\mathcal{DG}}
\newcommand{\GG}{\mathcal{G}}

\newcommand{\Def}{\mathrel{\mathop:}=}

\renewcommand{\emptyset}{\varnothing}

\newcommand{\oldcomment}[1]{}

\newcommand{\report}[1]{#1}
\newcommand{\paper}[1]{}

\DeclareMathOperator{\dom}{dom}

\newcommand{\id}{\mathsf{id}}

\crefname{algorithm}{alg.}{algorithms}%
\crefname{equation}{eq.}{equations}%
\crefname{chapter}{chapter}{chapters}%
\crefname{section}{sect.}{sections}%
\crefname{appendix}{app.}{appendices}%
\crefname{enumi}{item}{items}%
\crefname{footnote}{footnote}{footnotes}%
\crefname{figure}{fig.}{figures}%
\crefname{table}{table}{tables}%
\crefname{theorem}{thm.}{theorems}%
\crefname{lemma}{lemma}{lemmas}%
\crefname{corollary}{cor.}{corollaries}%
\crefname{proposition}{proposition}{propositions}%
\crefname{definition}{def.}{definitions}%
\crefname{result}{result}{results}%
\crefname{example}{ex.}{examples}%
\crefname{remark}{remark}{remarks}%
\crefname{note}{note}{notes}%
\crefname{lstlisting}{listing}{listings}%
\crefname{requirement}{req.}{requirements}%

\title{Integrating Loop Acceleration into Bounded Model Checking\thanks{funded by
    the Deutsche Forschungsgemeinschaft (DFG, German Research Foundation)
    - 235950644 (Project GI 274/6-2)}}
\titlerunning{Integrating Loop Acceleration into Bounded Model Checking}
\author{Florian Frohn$^{(\href{mailto:florian.frohn@informatik.rwth-aachen.de}{\mbox{\Letter}})}$\orcidID{0000-0003-0902-1994} and Jürgen Giesl$^{(\href{mailto:giesl@informatik.rwth-aachen.de}{\mbox{\Letter}})}$\orcidID{0000-0003-0283-8520}}
\institute{RWTH Aachen University,  Aachen, Germany\\
\email{\{florian.frohn,giesl\}@informatik.rwth-aachen.de}}
\authorrunning{F.\ Frohn, J.\ Giesl}

\paper{
\newif\ifbadgeavailable\newif\ifbadgefunctional\newif\ifbadgereusable
\badgeavailabletrue\badgefunctionalfalse\badgereusabletrue
\RequirePackage{graphicx}
\usepackage[firstpageonly=true,angle=0,vpos=.147\paperheight,hpos=.393\linewidth,vanchor=t,hanchor=l]{draftwatermark}
\SetWatermarkText{%
\raisebox{-3cm}
{\ifbadgeavailable\includegraphics[width=11mm]{badges/FM_2024_AE_available}\hspace{.815\linewidth}\else\hspace{.905\linewidth}\fi%
\ifbadgefunctional\includegraphics[width=11mm]{badges/FM_2024_AE_functional}\else\ifbadgereusable\includegraphics[width=11mm]{badges/FM_2024_AE_reusable}\fi\fi}}
}

\begin{document}

\renewcommand{\thelstlisting}{\arabic{lstlisting}}

\maketitle

\begin{abstract}
  \emph{Bounded Model Checking} (BMC) is a powerful technique for proving unsafety.
  However, finding \emph{deep counterexamples} that require a large bound is challenging for BMC.
  On the other hand, \emph{acceleration techniques} compute ``shortcuts'' that
  ``compress'' many execution steps into a single one.
  In this paper, we tightly integrate acceleration techniques into SMT-based bounded model checking.
  By adding suitable ``shortcuts'' on the fly, our approach can quickly detect
  deep counterexamples.
  Moreover, using so-called \emph{blocking clauses}, our approach can prove safety of examples where BMC diverges.
  An empirical comparison with other state-of-the-art techniques shows that our approach is highly competitive for proving unsafety, and orthogonal to existing techniques for proving safety.
\end{abstract}

\section{Introduction}
\label{sec:intro}

\emph{Bounded Model Checking} (BMC) is a powerful technique for disproving safety properties of, e.g., software or hardware systems.
However, as it uses breadth-first search to find counterexamples, the search space grows exponentially w.r.t.\ the \emph{bound}, i.e., the limit on the length of potential counterexamples.
Thus, finding \emph{deep counterexamples} that require large bounds is challenging for BMC.
On the other hand, \emph{acceleration techniques} can compute a first-order formula that characterizes the transitive closure of the transition relation induced by a loop.
Intuitively, such a formula corresponds to a ``shortcut'' that ``compresses'' many execution steps into a single one.
In this paper, we consider relations defined by quantifier-free first-order formulas over some background theory like non-linear integer arithmetic
and two disjoint vectors of variables $\vec{x}$ and $\vec{x}'$, called the \emph{pre-} and \emph{post-variables}.
Such \emph{transition formulas} can easily represent, e.g., \emph{transition systems}
(TSs), linear \emph{Constrained Horn Clauses} (CHCs), and \emph{control-flow automata}
(CFAs).\footnote{To this end, it suffices to introduce one additional variable that 
represents the control-flow location (for TSs and CFAs) or the predicate (for linear
CHCs).}
Thus, they subsume many popular intermediate representations used for verification of
systems specified in more expressive languages.

In contrast to, e.g., source code,
transition formulas are completely unstructured.
However, source code may be
unstructured, too (e.g., due to {\tt goto}s), i.e., one cannot rely on the input being well structured.
So the fact that our approach is independent from the structure of the input makes it broadly applicable.

\begin{example}
  \label{ex:ex1}
  Consider the transition formula $\tau \Def \tau_{x<100} \lor \tau_{x=100}$ where

  \noindent
  \begin{minipage}{0.57\textwidth}
    \begin{align*}
      \tau_{x<100} & {} \Def x < 100 \land x' = x + 1 \land y' = y \hspace{1em} \text{and}\\
      \tau_{x=100} & {} \Def x = 100 \land x' = 0 \land y' = y + 1.
    \end{align*}
    It defines a relation $\to_{\tau}$ on $\ZZ \times \ZZ$ by relating the pre-variables $x$ and $y$ with the post-variables $x'$ and $y'$.
    So for all $c_x,c_y, c'_x, c'_y \in \ZZ$,
  \end{minipage}\hspace*{.12cm}
  \begin{minipage}{0.38\textwidth}
    \medskip
    \begin{lstlisting}[language=C, caption=Implementation of $\tau$, captionpos=b, frame=single, label=code]
while (x <= 100) {
  while (x < 100) x++;
  x = 0, y++;
}
    \end{lstlisting}
  \end{minipage}

  \vspace{-.2em}
  \noindent
  we have $(c_x,c_y) \to_{\tau} (c'_x, c'_y)$ iff $[x/c_x,y/c_y,x'/c'_x,y'/c'_y]$ is a model of $\tau$,
i.e., iff there is a step from a state with $x=c_x \land y=c_y$ to a state with $x=c'_x
\land y=c'_y$ in \Cref{code}. 
  To prove that an \emph{error state} satisfying $\psi_{\err} \Def y
  \geq 100$ is reachable from an \emph{initial state} which satisfies $\psi_{\init} \Def x \leq 0 \land y \leq 0$,
BMC has to unroll $\tau$ $10100$ times.

  Our new technique \emph{Accelerated} BMC (ABMC) uses the following \emph{acceleration}
  \begin{equation}
    n > 0 \land x + n \leq 100 \land x' = x + n \land y' = y \tag{\ensuremath{\tau_{\inner}^+}} \label{accel1}
  \end{equation}
  of $\tau_{x<100}$:
  As we have $(c_x,c_y) \to^+_{\tau_{x<100}} (c'_x, c'_y)$ iff
  $\ref{accel1}[x/c_x,y/c_y,x'/c'_x,y'/c'_y]$ is satisfiable, \ref{accel1} is a
  ``shortcut'' for many $\to_{\tau_{x<100}}$-steps.
  
To compute such a shortcut \ref{accel1} from the formula $\tau_{x<100}$, we use existing
acceleration techniques \cite{acceleration-calculus}.
  In the example above, $n$ serves as loop counter.
  Then the literal $x' = x + 1$ of $\tau_{x<100}$ gives rise to the recurrence equations
  $x^{(0)} = x$ and $x^{(n)} = x^{(n-1)} + 1$, which yield the closed form $x^{(n)} = x +
  n$, resulting in the literal $x' = x + n$ of \ref{accel1}.
  Thus, the literal $x + n \leq 100$ of \ref{accel1} is equivalent to $x^{(n-1)} < 100$.
  As $x$ is monotonically increasing (i.e., $\tau_{x<100}$ implies $x < x'$), $x^{(n-1)} <
  100$ implies $x^{(n-2)}< 100$, $x^{(n-3)}< 100$, \ldots, $x^{(0)} < 100$, i.e., the loop
  $\tau_{x<100}$ can indeed be executed $n$ times.

  So \ref{accel1} can simulate arbitrarily many steps with $\tau$ in a single step, as long as $x$ does not exceed $100$.
  Here, acceleration was applied to $\tau_{x<100}$, i.e., the projection of $\tau$ to the case $x < 100$, which corresponds to the \underline{\textsf{i}}nner loop of \Cref{code}.
  We also call such projections \emph{transitions}.
  Later, ABMC also accelerates the \underline{\textsf{o}}uter loop (consisting of
  $\tau_{x=100}$, $\tau_{x<100}$,  and \ref{accel1}), resulting in
  \begin{equation}
    n > 0 \land x = 100 \land 1 < x' \leq 100 \land y' = y + n. \tag{\ensuremath{\tau_{\outer}^+}} \label{accel2}
  \end{equation}
  For technical reasons, our algorithm accelerates $[\tau_{x=100}, \tau_{x<100}, \ref{accel1}]$ instead of just $[\tau_{x=100}, \ref{accel1}]$, so that \ref{accel2} requires $1 < x'$, i.e., it only covers cases where $\tau_{x<100}$ is applied at least twice after $\tau_{x=100}$.
  Details will be clarified in \Cref{sec:abmc}, see in particular \Cref{fig:abmc}.
  Using these shortcuts, ABMC can prove unsafety with bound~$7$.
\end{example}
While our main goal is to improve BMC's capability to find deep counterexamples, the following straightforward observations can be used to \emph{block} certain parts of the transition relation in ABMC:
\begin{enumerate}
\item After accelerating a sequence of transitions, the resulting accelerated transition should be preferred over that sequence of
  transitions.
\item If an accelerated transition has been used, then the corresponding sequence of transitions should not be used immediately afterwards.
\end{enumerate}
Both observations exploit that an accelerated transition describes the transitive closure of the relation induced by the corresponding sequence of transitions.
Due to its ability to block parts of the transition relation, ABMC is able to prove safety
in cases where BMC would unroll the transition relation indefinitely.

\vspace*{-.2cm}

\subsubsection{Outline}

After introducing preliminaries in \Cref{sec:preliminaries}, we show how to use acceleration in
order to improve the BMC algorithm to ABMC in \Cref{sec:From BMC to ABMC}.
To increase ABMCs capabilities for proving safety,
\Cref{sec:blocking} refines ABMC by integrating
blocking clauses. In \Cref{sec:related}, we discuss related work, and in \Cref{sec:experiments}, we evaluate our
implementation of ABMC in our tool \tool{LoAT}.

\section{Preliminaries}
\label{sec:preliminaries}

We assume familiarity with basics from many-sorted first-order logic \cite{enderton}.
Without loss of generality, we assume that all formulas are in negation normal form (NNF).
$\VV$ is a countably infinite set of variables and $\AA$ is a first-order theory over a $k$-sorted signature $\Sigma$ with carrier $\CC = (\CC_{1},\ldots,\CC_{k})$.
For each entity $e$, $\VV(e)$ is the set of variables that occur in $e$.
$\QF(\Sigma)$ denotes the set of all quantifier-free first-order formulas over $\Sigma$, and
$\QF_\land(\Sigma)$ only contains conjunctions of $\Sigma$-literals.
We let $\top$ and $\bot$ stand for ``true'' and ``false'', respectively.

Given $\psi \in \QF(\Sigma)$ with $\VV(\psi) = \vec{y}$, we say that $\psi$ is $\AA$-\emph{valid} (written $\models_\AA \psi$) if every model of $\AA$ satisfies the universal closure $\forall \vec{y}.\ \psi$ of $\psi$.
Moreover, $\sigma: \VV(\psi) \to \CC$ is an $\AA$-\emph{model} of $\psi$ (written $\sigma \models_\AA \psi$) if $\models_\AA \sigma(\psi)$, where $\sigma(\psi)$ results from $\psi$ by instantiating all variables according to $\sigma$.
If $\psi$ has an $\AA$-model, then $\psi$ is $\AA$-\emph{satisfiable}.
We write $\psi \models_\AA \psi'$ for $\models_\AA (\psi \implies \psi')$, and $\psi \equiv_\AA
\psi'$ means $\models_\AA (\psi \iff \psi')$.
In the sequel, we omit the subscript $\AA$, and we just say ``valid'', ``model'', and ``satisfiable''.
We assume that $\AA$ is complete, i.e., we either have $\models \psi$ or $\models \neg
\psi$ for every closed formula 
over $\Sigma$.

We write $\vec{x}$ for sequences and $x_i$ is the $i^{th}$ element of $\vec{x}$.
We use ``$\concat$'' for concatenation of sequences, where we identify sequences of length $1$ with their elements, so we may write, e.g., $x\concat\mathit{xs}$ instead of $[x]\concat\mathit{xs}$.

Let $d \in \NN$ be fixed, and let $\vec{x},\vec{x}' \in \VV^d$ be disjoint vectors of
pairwise different variables, called the \emph{pre-} and \emph{post-variables}.
Each $\tau \in \QF(\Sigma)$ induces a \emph{transition relation} $\to_\tau$ on $\CC^d$ where $\vec{s} \to_\tau \vec{t}$ iff $\tau[\vec{x}/\vec{s},\vec{x}'/\vec{t}]$ is satisfiable.
Here, $[\vec{x}/\vec{s},\vec{x}'/\vec{t}]$ denotes the substitution $\theta$ with
$\theta(x_i) = s_i$ and $\theta(x'_i) = t_i$ for all $1 \leq i \leq d$.
We refer to elements of $\QF(\Sigma)$ as \emph{transition formulas} whenever we are interested in their induced transition relation.
Moreover, we also refer to \emph{conjunctive} transition formulas (i.e., elements of $\QF_\land(\Sigma)$) as \emph{transitions}.
A \emph{safety problem} $\TT$ is a triple $(\psi_{\init}, \tau,\psi_{\err}) \in \QF(\Sigma) \times \QF(\Sigma) \times \QF(\Sigma)$ where $\VV(\psi_{\init}) \cup \VV(\psi_{\err}) \subseteq \vec{x}$.
It is \emph{unsafe} if there are $\vec{s},\vec{t} \in \CC^d$ such that $[\vec{x} / \vec{s}] \models \psi_\init$, $\vec{s} \to^*_\tau \vec{t}$, and $[\vec{x} / \vec{t}] \models \psi_\err$.

The \emph{composition} of $\tau$ and $\tau'$ is $\compose(\tau,\tau') \Def \tau[\vec{x}' / \vec{x}''] \land \tau'[\vec{x} / \vec{x}'']$ where $\vec{x}'' \in \VV^d$ is fresh.
Here, we assume $\VV(\tau) \cap \VV(\tau') \subseteq \vec{x} \cup \vec{x}'$ (which can be ensured by renaming other variables correspondingly).
So ${\to_{\compose(\tau,\tau')}} = {\to_\tau} \circ {\to_{\tau'}}$ (where ${\circ}$
denotes relational composition).
For finite sequences of transition formulas
we define $\compose([]) \Def (\vec{x} = \vec{x}')$ (i.e., $\to_{\compose([])}$ is the identity relation) and
$\compose(\tau\concat\vec{\tau}) \Def \compose(\tau,\compose(\vec{\tau}))$.
We abbreviate ${\to_{\compose(\vec{\tau})}}$ by ${\to_{\vec{\tau}}}$.

\emph{Acceleration techniques} compute the transitive closure of relations.
In the following definition, we only consider conjunctive transition formulas, since many
existing acceleration techniques do not support disjunctions \cite{bozga09a}, or
approximate in the presence 
of disjunctions \cite{acceleration-calculus}.
So the restriction to conjunctive formulas ensures that our approach works with arbitrary
existing acceleration techniques.

\begin{definition}[Acceleration]
  \label{def:accel}
  An \emph{acceleration technique} is a
  function $\accel: \QF_\land(\Sigma) \to \QF_\land(\Sigma')$ such that ${\to_{\accel(\tau)}} \subseteq {\to_{\tau}^+}$, where $\Sigma'$ is the signature of a first-order theory $\AA'$.
\end{definition}
We abbreviate $\accel(\compose(\vec{\tau}))$ by $\accel(\vec{\tau})$.
So as we aim at finding counterexamples, we allow under-approximating acceleration techniques, i.e., we do not require ${\to_{\accel(\tau)}} = {\to_{\tau}^+}$.
\Cref{def:accel} allows $\AA' \neq \AA$, as most theories are not ``closed under acceleration''.
For example, accelerating the following Presburger formula on the left may yield the non-linear formula on the right:
\[
  x' = x + y \land y' = y \qquad \qquad n > 0 \land x' = x + n \cdot y \land y' = y.
\]

\section{From BMC to ABMC}\label{sec:From BMC to ABMC}

In this section, we introduce accelerated bounded model checking.
To this end, we first recapitulate bounded model checking in \Cref{sec:bmc}.
Then we present ABMC in \Cref{sec:abmc}.
To implement ABMC efficiently, heuristics to decide when to perform acceleration are needed.
Thus, we present such a heuristic in \Cref{sec:tuning}.

\subsection{Bounded Model Checking}
\label{sec:bmc}

\begin{algorithm}[t]
  \caption{$\BMC$ -- Input: a safety problem $\TT = (\psi_\init,\tau,\psi_\err)$}\label{alg:bmc}
  $b \gets 0; \quad \add(\mu_b(\psi_\init))$\; \label{bmc:init}
  \While{$\top$}{
    $\push()$; \quad $\add(\mu_b(\psi_\err))$\;\label{bmc:err1}
    \leIf{\label{bmc:err2}$\checksat()$}{
      \Return{$\unsafe$}
    }{
      $\pop(); \quad \add(\mu_b(\tau))$ \label{bmc:pop}
    }
    \leIf{\label{bmc:safe}$\neg\checksat()$}{
      \Return{$\safe$}
    }{
      $b \gets b + 1$\label{bmc:inc}
    }
  }
\end{algorithm}
\Cref{alg:bmc} shows how to implement BMC on top of an incremental SMT solver.
In Line~\ref{bmc:init}, the description of the initial states is added to the SMT problem.
Here and in the following, 
for all $i \in \NN$ we define $\mu_i(x) \Def x^{(i)}$ if $x \in \VV \setminus \vec{x}'$,
and $\mu_i(x') = x^{(i+1)}$ if $x' \in \vec{x}'$. 
So in particular, we have  $\mu_i(\vec{x}) = \vec{x}^{(i)}$
and $\mu_i(\vec{x}') = \vec{x}^{(i+1)}$, where
we assume that $\vec{x}^{(0)},\vec{x}^{(1)},\ldots \in \VV^d$ are
disjoint vectors of pairwise different variables.
In the loop, we set a backtracking point with the ``$\push()$'' command and add a suitably
variable-renamed version of the description of the error states to the SMT problem in Line~\ref{bmc:err1}.
Then we check for satisfiability to see if an error state is reachable with the current bound in Line~\ref{bmc:pop}.
If this is not the case, the description of the error states is removed with the ``$\pop()$'' command that deletes all formulas from the SMT problem that have been added since the last backtracking point.
Then a variable-renamed version of the transition formula $\tau$ is added to the SMT problem.
If this results in an unsatisfiable problem in Line~\ref{bmc:safe}, then the whole search space has been exhausted, i.e., then $\TT$ is safe.
Otherwise, we enter the next iteration.

\begin{example}[BMC]
  \label{ex:bmc}
    For the first $100$ iterations of
    \Cref{alg:bmc} on  \Cref{ex:ex1}, all models found in Line~\ref{bmc:inc} satisfy the $1^{st}$ disjunct
  $\mu_b(\tau_{x<100})$ of $\mu_b(\tau)$. 
  Then we may have $x^{(100)} = 100$, so that the $2^{nd}$ disjunct $\mu_b(\tau_{x=100})$ of $\mu_b(\tau)$ applies once and we get $y^{(101)} = y^{(100)} + 1$.
  After another $100$ iterations, the $2^{nd}$ disjunct $\mu_b(\tau_{x=100})$ may apply
  again, etc.
  After $100$ applications of the $2^{nd}$ disjunct (and thus a total of $10100$ steps), there is a model with $y^{(10100)} = 100$, so that unsafety is proven.
\end{example}

\subsection{Accelerated Bounded Model Checking}
\label{sec:abmc}

To incorporate acceleration into BMC, we have to bridge the gap between (disjunctive) transition formulas and acceleration techniques, which require conjunctive transition formulas.
To this end, we use \emph{syntactic implicants}.
\begin{definition}[Syntactic Implicant Projection \cite{adcl}]
  Let $\tau \in \QF(\Sigma)$ be in NNF and assume $\sigma \models \tau$.
  We define the \emph{syntactic implicants} $\sip(\tau)$ of $\tau$ as follows:
  \begin{align*}
    \sip(\tau,\sigma) \Def \bigwedge \{\lambda  \mid \lambda \text{ is a literal of } \tau,
    \, \sigma \models \lambda\} \qquad
    \sip(\tau) \Def \{\sip(\tau,\sigma) \mid \sigma \models \tau\}
  \end{align*}
\end{definition}
Since $\tau$ is in NNF, $\sip(\tau,\sigma)$ implies $\tau$, and it is easy to see that
$\tau \equiv \bigvee \sip(\tau)$.
Whenever the call to the SMT solver in Line~\ref{bmc:safe} of \Cref{alg:bmc} yields $\sat$, the resulting model gives rise to a sequence of syntactic implicants, called the \emph{trace}.
To define the trace formally, note that when we integrate acceleration into $\BMC$, we may
not only add $\tau$ to the SMT formula as in Line~\ref{bmc:pop}, but also \emph{learned transitions} that result from acceleration.
Thus, the following definition allows for changing the transition formula.
In the sequel, $\circ$ also
denotes composition of substitutions, i.e., $\theta' \circ \theta \Def [x / \theta'(\theta(x)) \mid x \in \dom(\theta') \cup \dom(\theta)]$.

\begin{definition}[Trace]\label{def:trace}
  Let $[\tau_i]_{i=0}^{b-1}$ be a sequence of transition formulas and let $\sigma$ be a model of $\bigwedge_{i=0}^{b-1} \mu_i(\tau_i)$.
  Then the \emph{trace induced by $\sigma$} 
  is
  \[
    \trace_b(\sigma,[\tau_i]_{i=0}^{b-1}) \Def [\sip(\tau_i, \sigma \circ \mu_i)]_{i=0}^{b-1}.
  \]
  We write $\trace_b(\sigma)$ instead of $\trace_b(\sigma,[\tau_i]_{i=0}^{b-1})$ if $[\tau_i]_{i=0}^{b-1}$ is clear from the context.
\end{definition}
So each model $\sigma$ of $\bigwedge_{i=0}^{b-1} \mu_i(\tau_i)$ corresponds to a sequence
of steps with the relations $\to_{\tau_0}, \to_{\tau_1}, \ldots, \to_{\tau_{b-1}}$, and the trace
induced by $\sigma$ contains the syntactic implicants of the formulas $\tau_i$ that were used in this
sequence.
\begin{example}[Trace]
  Reconsider \Cref{ex:bmc}.
  After two iterations of the loop of \Cref{alg:bmc}, the SMT problem consists of the following
  formulas:

  \noindent
  \resizebox{\textwidth}{!}{
    \begin{minipage}{1.12\textwidth}
      \begin{align*}
        &x^{(0)} \leq 0 \land y^{(0)} \leq 0 \tag{$\psi_\init$} \\
        &
          (x^{(0)} < 100 \land x^{(1)} = x^{(0)} + 1 \land y^{(1)} = y^{(0)})
          \lor (x^{(0)} = 100 \land x^{(1)} = 0 \land y^{(1)} = y^{(0)} + 1) \tag{$\tau$}\\
        &
          (x^{(1)} < 100 \land x^{(2)} = x^{(1)} + 1 \land y^{(2)} = y^{(1)})
          \lor (x^{(1)} = 100 \land x^{(2)} = 0 \land y^{(2)} = y^{(1)} + 1) \tag{$\tau$}
      \end{align*}
    \end{minipage}
  }

  \vspace{\belowdisplayskip}
  \noindent
  With $\sigma = [x^{(i)}/i,y^{(i)}/0 \mid 0 \leq i \leq 2]$, we get $\trace_2(\sigma) = [\tau_{x<100},\tau_{x<100}]$, as:
  \begin{align*}
    \sip(\tau, \sigma \circ \mu_0) & {} = \sip(\tau, [x/0,y/0,x'/1,y'/0]) = \tau_{x<100} \\
    \sip(\tau, \sigma \circ \mu_1) & {} = \sip(\tau, [x/1,y/0,x'/2,y'/0]) = \tau_{x<100}
  \end{align*}
\end{example}
To detect situations where applying acceleration techniques pays off, we need to distinguish traces that contain loops from non-looping ones.
Since transition formulas are unstructured, the usual techniques for detecting loops
(based on, e.g., program syntax or control flow graphs) do not apply in our
setting.
Instead, we rely on the \emph{dependency graph} of the transition formula.

\begin{definition}[Dependency Graph]
  Let $\tau$ be a transition formula.
  Its \emph{dependency graph} $\DG = (V,E)$ is a directed graph whose vertices $V \Def \sip(\tau)$ are $\tau$'s syntactic implicants, and $\tau_1 \to \tau_2 \in E$ if $\compose(\tau_1, \tau_2)$ is satisfiable.
  We say that $\vec{\tau} \in \sip(\tau)^c$ is $\DG$-\emph{cyclic} if $c>0$ and $(\tau_1 \to
  \tau_{2}),\ldots,(\tau_{c-1} \to \tau_{c}),(\tau_c \to \tau_1) \in E$.
\end{definition}

\noindent
So intuitively, the syntactic implicants correspond to the different cases of $\to_\tau$, 
\begin{wrapfigure}[3]{r}{5cm}
  \vspace*{-.2cm}
  \hspace*{-.5cm}  \begin{tikzpicture}[bend angle=15]
      \node[draw=black] (A) at (0,0) {$\tau_{x<100}$};
      \node[draw=black] (B) at (3,0) {$\tau_{x=100}$};

      \path [->] (A) edge [loop left] node[left] {} (A);
      \path [->] (A.north east) edge [bend left] node[left] {} (B.north west);
      \path [->] (B.south west) edge [bend left] node[left] {} (A.south east);
     \end{tikzpicture}
     \end{wrapfigure}
and $\tau$'s dependency graph corresponds to the control flow graph of $\to_\tau$.
The dependency graph for \Cref{ex:ex1} is on the side.

However, as the size of $\sip(\tau)$ is worst-case exponential in the number of disjunctions in $\tau$, we do not compute $\tau$'s dependency graph eagerly.
Instead, ABMC maintains an under-approximation, i.e., a subgraph $\GG$ of the dependency
graph, which is extended whenever two transitions that are not yet connected by an edge
occur consecutively on the trace.
As soon as a $\GG$-cyclic suffix $\vec{\tau}^\circlearrowleft$ is detected on the trace, we may accelerate it.
Therefore, the trace may also contain the learned transition $\accel(\vec{\tau}^\circlearrowleft)$ in subsequent iterations.
Hence, to detect cyclic suffixes that contain learned transitions, they have to be represented in $\GG$ as well.
Thus, $\GG$ is in fact a subgraph of the dependency graph of $\tau \lor \bigvee \LL$, where $\LL$ is the set of all transitions that have been learned so far.

\begin{algorithm}[t]
  \caption{$\ABMC$ -- Input: a safety problem $\TT = (\psi_\init,\tau,\psi_\err)$}\label{alg:abmc}
  $b \gets 0; \quad V \gets \emptyset; \quad E \gets \emptyset; \quad \add(\mu_b(\psi_\init))$\; \label{abmc:init0}
  \leIf{\label{abmc:init2}$\neg\checksat()$}{
    \Return{$\safe$}
  }{
    $\sigma \gets \getmodel()$
  }
  \While{$\top$}{
     $\push()$; $\add(\mu_b(\psi_\err))$\; \label{abmc:err1}
    \leIf{\label{abmc:err2}$\checksat()$}{
      \Return{$\unsafe$}
    }{
      $\pop()$
    }
    $\vec{\tau} \gets \trace_b(\sigma); \quad V \gets V \cup \vec{\tau}; \quad E \gets E \cup \{(\tau_1,\tau_2) \mid [\tau_1,\tau_2] \text{ is an infix of } \vec{\tau}\}$\; \label{abmc:model}
    \lIf{\label{abmc:rec}$\vec{\tau} = \vec{\pi}\concat\vec{\pi}^\circlearrowleft \land \vec{\pi}^\circlearrowleft \text{ is cyclic} \land \shouldaccel(\vec{\pi}^\circlearrowleft)$}{
      $\add(\mu_b(\tau \lor \accel(\vec{\pi}^\circlearrowleft)))$ \label{abmc:add}
    }
    \lElse{$\add(\mu_b(\tau))$}\label{abmc:deepen}
    \leIf{\label{abmc:safe}$\neg\checksat()$}{
      \Return{$\safe$}
    } {
      $\sigma \gets \getmodel(); \quad b \gets b + 1$
    }
  }
  \end{algorithm}
This gives rise to the ABMC algorithm, which is shown in \Cref{alg:abmc}.
Here, we just write ``cyclic'' instead of $(V,E)$-cyclic.
The difference to \Cref{alg:bmc} can be seen in Lines~\ref{abmc:model} and \ref{abmc:add}.
In Line~\ref{abmc:model}, the trace is constructed from the current model.
Then, the approximation of the dependency graph is refined such that it contains vertices
for all elements of
the trace, and edges for consecutive elements of the trace.
In Line~\ref{abmc:rec}, a cyclic suffix of the trace may get accelerated, provided that the call to $\shouldaccel$ (which will be discussed
in detail in \Cref{sec:tuning}) returns $\top$.
In this way, in the next iteration the SMT solver can choose a model that satisfies $\accel(\vec{\pi}^\circlearrowleft)$ and thus
simulates several instead of just one $\to_\tau$-step.
Note, however, that we do \emph{not} update $\tau$ with $\tau \lor \accel(\vec{\pi}^\circlearrowleft)$.
So in every iteration, at most one learned transition is added to the SMT problem.
In this way, we avoid blowing up $\tau$ unnecessarily.
Note that we only accelerate ``real'' cycles $\vec{\pi}^\circlearrowleft$ where
$\compose(\vec{\pi}^\circlearrowleft)$ is satisfiable, since $\vec{\pi}^\circlearrowleft$ is a suffix of the trace, whose satisfiability is witnessed by $\sigma$.

As we rely on syntactic implicants and dependency graphs to detect cycles, $\ABMC$ is decoupled from the specific encoding of the input.
So for example, 
transition formulas may be represented in CNF, DNF,
or any other structure.

\Cref{fig:abmc} shows a run of \Cref{alg:abmc} on \Cref{ex:ex1}, where the formulas that
are added to the SMT problem are highlighted in \colorbox{mygray}{gray},
and $x^{(i)} \mapsto c$ abbreviates $\sigma(x^{(i)}) = c$.
For simplicity, we assume that $\shouldaccel$ always returns $\top$, and
the model $\sigma$ is only extended in each step, i.e.,  $\sigma(x^{(i)})$ and
$\sigma(y^{(i)})$ remain unchanged for all $0 \leq i < b$.
In general, the SMT solver can choose different values for $\sigma(x^{(i)})$ and $\sigma(y^{(i)})$ in every iteration.
On the right, we show the current bound $b$, and the formulas that give rise to the formulas on the left when renaming their variables suitably with $\mu_b$.
Initially, the approximation $\GG = (V,E)$ of the dependency graph is empty.
When $b=2$, the trace is $[\tau_{x < 100},\tau_{x < 100}]$, and the corresponding edge is added to $\GG$.
Thus, the trace has the cyclic suffix $\tau_{x < 100}$ and we accelerate it, resulting in \ref{accel1}, which is added to the SMT problem.
Then we obtain the trace $[\tau_{x < 100},\tau_{x < 100},\ref{accel1}]$, and the edge $\tau_{x < 100} \to \ref{accel1}$ is added to $\GG$.
Note that \Cref{alg:abmc} does not enforce the use of $\ref{accel1}$, so $\tau$ might still be unrolled instead, depending on the models found by the SMT solver.
We will address this issue in \Cref{sec:blocking}.

Next, $\tau_{x = 100}$ already applies with $b=4$ (whereas it only applied with $b=100$ in \Cref{ex:bmc}).
So the trace is $[\tau_{x < 100},\tau_{x < 100},\ref{accel1},\tau_{x = 100}]$, and the edge $\ref{accel1} \to \tau_{x = 100}$ is added to $\GG$.
Then we obtain the trace $[\tau_{x < 100},\tau_{x < 100},\ref{accel1},\tau_{x = 100},\tau_{x < 100}]$, and add $\tau_{x = 100} \to \tau_{x < 100}$ to $\GG$.
Since the suffix $\tau_{x < 100}$ is again cyclic, we accelerate it and add $\ref{accel1}$
to the SMT problem.
After one more step, the trace $[\tau_{x < 100},\tau_{x < 100},\ref{accel1},\tau_{x = 100},\tau_{x < 100},\ref{accel1}]$ has the cyclic suffix $[\tau_{x = 100},\tau_{x < 100},\ref{accel1}]$.
Accelerating it yields $\ref{accel2}$, which is added to the SMT problem.
Afterwards, unsafety can be proven with $b=7$.

Since using acceleration is just a heuristic to speed up BMC, all basic properties of BMC
immediately carry over to ABMC.
\begin{theorem}[Properties of ABMC]
  ABMC is
  \begin{description}
  \item[Sound:] If $\ABMC(\TT)$ returns $\mathsf{(un)safe}$, then $\TT$ is (un)safe.
  \item[Refutationally Complete:] If  $\TT$ is unsafe, then $\ABMC(\TT)$ returns $\unsafe$.
  \item[Non-Terminating:] If $\TT$ is safe, then $\ABMC(\TT)$ may not terminate.
  \end{description}
\end{theorem}

\begin{figure}[t]
    \begin{forest}
      for tree={%
        folder,
        grow'=0,
        fit=band,
        s sep=0mm,
        l sep=1em
      }
      [$\text{ABMC}(\TT)$
      [\protect{\begin{minipage}{0.87\textwidth}\ref{abmc:init0}: \colorbox{mygray}{$x^{(0)} \leq 0 \land y^{(0)} \leq 0$}\end{minipage} $\hfill \mathllap{\psi_\init, {}}b=0$}]
      [\protect{\begin{minipage}{0.87\textwidth}\parbox{2.5em}{\ref{abmc:init2} \& \ref{abmc:model}:}
            \parbox{0.27\textwidth}{$x^{(0)} \mapsto 0, y^{(0)} \mapsto 0$} $||$ \parbox{0.21\textwidth}{$\vec{\tau} \gets []$} $||$ \parbox{0.37\textwidth}{$E \gets \emptyset$}\end{minipage}}]
      [\protect{\begin{minipage}{0.87\textwidth}\ref{abmc:deepen}: \colorbox{mygray}{$(x^{(0)} < 100 \land x^{(1)} = x^{(0)} + 1 \land y^{(1)} = y^{(0)}) \lor \ldots$}\end{minipage} $\hfill \tau$}]
      [\protect{\begin{minipage}{0.87\textwidth}\parbox{2.5em}{\ref{abmc:model} \& \ref{abmc:safe}:}
            \parbox{0.27\textwidth}{$x^{(1)} \mapsto 1, y^{(1)} \mapsto 0$} $||$ \parbox{0.21\textwidth}{$\vec{\tau} \gets [\tau_{x<100}]$} $||$ \parbox{0.37\textwidth}{$E \gets \emptyset$}\end{minipage} \hfill $b=1$}]
      [\protect{\begin{minipage}{0.87\textwidth}\ref{abmc:deepen}: \colorbox{mygray}{$(x^{(1)} < 100 \land x^{(2)} = x^{(1)} + 1 \land y^{(2)} = y^{(1)}) \lor \ldots$}\end{minipage} $\hfill \tau$}]
      [\protect{\begin{minipage}{0.87\textwidth}\parbox{2.5em}{\ref{abmc:model} \& \ref{abmc:safe}:}
            \parbox{0.27\textwidth}{$x^{(2)} \mapsto 2, y^{(2)} \mapsto 0$} $||$ \parbox{0.21\textwidth}{$\vec{\tau} \gets \vec{\tau}::\tau_{x<100}$} $||$ \parbox{0.37\textwidth}{$E \gets \{\tau_{x<100} \to \tau_{x<100}\}$}\end{minipage} \hfill $b=2$}]
      [\protect{\begin{minipage}{0.87\textwidth}\ref{abmc:add}: \colorbox{mygray}{$\ldots \lor (n^{(2)} > 0 \land x^{(2)} + n^{(2)} \leq 100 \land x^{(3)} = x^{(2)} + n^{(2)} \land y^{(3)} = y^{(2)})$}\end{minipage} $\hfill \tau \lor \ref{accel1}$}]
      [\protect{\begin{minipage}{0.87\textwidth}\parbox{2.5em}{\ref{abmc:model} \& \ref{abmc:safe}:}
            \parbox{0.27\textwidth}{$x^{(3)} \mapsto 100, y^{(3)} \mapsto 0$} $||$ \parbox{0.21\textwidth}{$\vec{\tau} \gets \vec{\tau}::\ref{accel1}$} $||$ \parbox{0.37\textwidth}{$E \gets E \cup \{\tau_{x<100} \to \ref{accel1}\}$}\end{minipage} \hfill $b=3$}]
      [\protect{\begin{minipage}{0.87\textwidth}\ref{abmc:deepen}: \colorbox{mygray}{$\ldots \lor (x^{(3)} = 100 \land x^{(4)} = 0 \land y^{(4)} = y^{(3)} + 1)$}\end{minipage} $\hfill \tau$}]
      [\protect{\begin{minipage}{0.87\textwidth}\parbox{2.5em}{\ref{abmc:model} \& \ref{abmc:safe}:}
            \parbox{0.27\textwidth}{$x^{(4)} \mapsto 0, y^{(4)} \mapsto 1$} $||$ \parbox{0.21\textwidth}{$\vec{\tau} \gets \vec{\tau}::\tau_{x=100}$} $||$ \parbox{0.37\textwidth}{$E \gets E \cup \{\ref{accel1} \to \tau_{x=100}\}$}\end{minipage} \hfill $b=4$}]
      [\protect{\begin{minipage}{0.87\textwidth}\ref{abmc:deepen}: \colorbox{mygray}{$(x^{(4)} < 100 \land x^{(5)} = x^{(4)} + 1 \land y^{(5)} = y^{(4)}) \lor \ldots$}\end{minipage} $\hfill \tau$}]
      [\protect{\begin{minipage}{0.87\textwidth}\parbox{2.5em}{\ref{abmc:model} \& \ref{abmc:safe}:}
            \parbox{0.27\textwidth}{$x^{(5)} \mapsto 1, y^{(5)} \mapsto 1$} $||$ \parbox{0.21\textwidth}{$\vec{\tau} \gets \vec{\tau}::\tau_{x<100}$} $||$ \parbox{0.37\textwidth}{$E \gets E \cup \{\tau_{x=100} \to \tau_{x < 100}\}$}\end{minipage} \hfill $b=5$}]
      [\protect{\begin{minipage}{0.87\textwidth}\ref{abmc:add}: \colorbox{mygray}{$\ldots \lor (n^{(5)} > 0 \land x^{(5)} + n^{(5)} \leq 100 \land x^{(6)} = x^{(5)} + n^{(5)} \land y^{(6)} = y^{(5)})$}\end{minipage} $\hfill \tau \lor \ref{accel1}$}]
      [\protect{\begin{minipage}{0.87\textwidth}\parbox{2.5em}{\ref{abmc:model} \& \ref{abmc:safe}:}
            \parbox{0.27\textwidth}{$x^{(6)} \mapsto 100, y^{(6)} \mapsto 1$} $||$ \parbox{0.21\textwidth}{$\vec{\tau} \gets \vec{\tau}::\ref{accel1}$} $||$ \parbox{0.37\textwidth}{$E \gets E$}\end{minipage} \hfill $b=6$}]
      [\protect{\begin{minipage}{0.87\textwidth}\ref{abmc:add}: \colorbox{mygray}{$\ldots \lor (n^{(6)} > 0 \land x^{(6)} = 100 \land 1 < x^{(7)} \leq 100 \land y^{(7)} = y^{(6)} + n^{(6)})$}\end{minipage} $\hfill \tau \lor \ref{accel2}$}
      [\protect{\begin{minipage}{\dimexpr0.87\textwidth-1em}\ref{abmc:err1}: \colorbox{mygray}{$y^{(7)} \geq 100$}\end{minipage} \hfill $b=7$}]
      [\ref{abmc:err2}: \protect{$\unsafe$}]]
      ]
    \end{forest}
  \caption{Running ABMC on \Cref{ex:ex1}}
  \label{fig:abmc}
\end{figure}

\subsection{Fine Tuning Acceleration}
\label{sec:tuning}

We now discuss $\shouldaccel$, our heuristic for applying acceleration.
To explain the intuition of our heuristic, 
we assume that acceleration does not approximate and thus ${\to_{\accel(\vec{\tau})}} = {\to^+_{\vec{\tau}}}$, but in our implementation, we also use it if ${\to_{\accel(\vec{\tau})}} \subset {\to^+_{\vec{\tau}}}$.
This is uncritical for correctness, as using acceleration in \Cref{alg:abmc} is \emph{always} sound.

First, acceleration should be applied to cyclic suffixes consisting of a single
\emph{original} (i.e., non-learned) transition.
However, applying acceleration to a single learned transition is pointless, as
\[
  {\to_{\accel(\accel(\tau))}} = {\to_{\accel(\tau)}^+} = {(\to^+_{\tau})^+} = {\to^+_{\tau}} = {\to_{\accel(\tau)}}.
\]
\vspace{-1em}
\begin{requirement}
  \label{req2}
  $\shouldaccel([\pi]) = \top$ iff $\pi \in \sip(\tau)$. 
\end{requirement}
Next, for every cyclic sequence $\vec{\pi}$, we have
 \begin{align*}
   {\to_{\accel(\vec{\pi}::\accel(\vec{\pi}))}} = {\to^+_{\vec{\pi}::\accel(\vec{\pi})}} = ({\to_{\vec{\pi}} \circ \to_{\vec{\pi}}^+})^+  = {\to_{\vec{\pi}} \circ \to_{\vec{\pi}}^+} = {\to_{\vec{\pi}::\accel(\vec{\pi})}},
 \end{align*}
and thus accelerating $\vec{\pi}::\accel(\vec{\pi})$ is pointless, too.
More generally, we want to prevent acceleration of sequences $\vec{\pi}_2::\accel(\vec{\pi})::\vec{\pi}_1$ where $\vec{\pi} = \vec{\pi}_1::\vec{\pi}_2$
as
\[
  {\to_{\vec{\pi}_2::\accel(\vec{\pi})::\vec{\pi}_1}^2} = {\to_{\vec{\pi}_2::\accel(\vec{\pi})::\vec{\pi}::\accel(\vec{\pi})::\vec{\pi}_1}} \subseteq {\to_{\vec{\pi}_2::\accel(\vec{\pi})::\vec{\pi}_1}}
\]
and thus ${\to_{\accel(\vec{\pi}_2::\accel(\vec{\pi})::\vec{\pi}_1)}} =
{\to_{\vec{\pi}_2::\accel(\vec{\pi})::\vec{\pi}_1}^+} =
{\to_{\vec{\pi}_2::\accel(\vec{\pi})::\vec{\pi}_1}}$.
So in general, the cyclic suffix of such a trace consists of a cycle $\vec{\pi}$ and $\accel(\vec{\pi})$, but it does not necessarily start with either of them.
To take this into account, we rely on the notion of \emph{conjugates}.
\begin{definition}[Conjugate]
  We say that two vectors $\vec{v},\vec{w}$ are \emph{conjugates} (denoted $\vec{v} \equiv_\circ \vec{w}$) if $\vec{v}=\vec{v}_1 \concat \vec{v}_2$ and $\vec{w}=\vec{v}_2 \concat \vec{v}_1$.
\end{definition}
So a conjugate of a cycle corresponds to the same cycle with another entry point.$\!\!$
\begin{requirement}
  \label{req4}
  $\shouldaccel(\vec{\pi}') = \bot$ if $\vec{\pi}'
  \equiv_\circ \vec{\pi} \concat \accel(\vec{\pi})$ for some $\vec{\pi}$.
\end{requirement}
In general, however, we also want to accelerate cyclic suffixes that contain learned
transitions to deal with nested loops, as in the last acceleration step of \Cref{fig:abmc}.
\begin{requirement}
  \label{req5}
  $\shouldaccel(\vec{\pi}') = \top$ if $\vec{\pi}' \not\equiv_\circ \vec{\pi}::\accel(\vec{\pi})$ for all $\vec{\pi}$.
\end{requirement}
\Cref{req2,req4,req5} give rise to a complete specification
for $\shouldaccel$:
If the cyclic suffix is a singleton, the decision is made based on \Cref{req2},
and otherwise the decision is made based on \Cref{req4,req5}.
However, this specification misses one important case:
Recall that the trace was $[\tau_{x<100},\tau_{x<100}]$ before acceleration was applied for the first time in \Cref{fig:abmc}.
While both $[\tau_{x<100}]$ and
$[\tau_{x<100},\tau_{x<100}]$ are cyclic, the latter should not be accelerated, since the formula $\accel([\tau_{x<100},\tau_{x<100}])$ corresponds to a special case of \ref{accel1} that only represents an even number of steps with $\tau_{x<100}$.
Here, the problem is that the cyclic suffix contains a \emph{square}, i.e., two
adjacent repetitions of the same non-empty sub-sequence.
\begin{requirement}
  \label{req6}
  $\shouldaccel(\vec{\pi}) = \bot$ if $\vec{\pi}$ contains a square.
\end{requirement}
Thus, $\shouldaccel(\vec{\pi}')$ yields $\top$ iff the following holds:
\begin{align*}
  (|\vec{\pi}'| = 1 \land \vec{\pi}' \in \sip(\tau)) \lor (
                                   |\vec{\pi}'| > 1 \land \vec{\pi}' \text{ is square-free} \land \forall \vec{\pi}.\ (\vec{\pi}' \not\equiv_\circ \vec{\pi}\concat\accel(\vec{\pi})))
\end{align*}
All properties that are required to implement $\shouldaccel$ can easily be checked automatically.
To check $\vec{\pi}' \not\equiv_\circ \vec{\pi} \concat \accel(\vec{\pi})$, our implementation maintains a map from learned transitions to
the corresponding cycles that have been accelerated.

However, to implement \Cref{alg:abmc}, there is one more missing piece:
As the choice of the cyclic suffix 
in Line~\ref{abmc:rec} is non-deterministic, a heuristic for choosing it is required.
In our implementation, we choose the \emph{shortest} cyclic suffix such that $\shouldaccel$ returns $\top$.
The reason is that, as observed in \cite{adcl}, accelerating short cyclic suffixes before longer ones allows for learning more general transitions.

\section{Guiding ABMC with Blocking Clauses}
\label{sec:blocking}

As mentioned in \Cref{sec:abmc}, \Cref{alg:abmc} does not enforce the use of learned transitions.
Thus, depending on the models found by the SMT solver, ABMC may behave just like BMC.
We now improve ABMC by integrating \emph{blocking clauses} that prevent it from unrolling loops instead of using learned transitions.
Here, we again assume ${\to_{\accel(\vec{\tau})}} = {\to^+_{\vec{\tau}}}$, i.e., that acceleration does not approximate.
Otherwise, blocking clauses are only sound for proving unsafety, but not for proving safety.

Blocking clauses exploit the following straightforward observation:
If the learned transition $\tau_\ell = \accel(\vec{\pi}^\circlearrowleft)$ has been added to the SMT problem with bound $b$ and an error state can be
reached via a trace with prefix
\[
  \vec{\pi} = [\tau_0,\ldots,\tau_{b-1}]::\vec{\pi}^\circlearrowleft \qquad \text{or} \qquad \vec{\pi}' = [\tau_0,\ldots,\tau_{b-1},\tau_\ell]::\vec{\pi}^\circlearrowleft,
  \]
then an error state can also be reached via a trace with the prefix
$[\tau_0,\ldots,\tau_{b-1},\tau_\ell]$, which
is not continued with $\vec{\pi}^\circlearrowleft$.
Thus, we may remove traces of the form $\vec{\pi}$ and $\vec{\pi}'$ from the search space by modifying the SMT problem accordingly.

To do so, we assign a unique identifier to each learned transition, and we introduce a
fresh integer-valued variable $\ell$ which is set to the corresponding identifier whenever
a learned transition is used, and to $0$, otherwise.
\begin{example}[Blocking Clauses]
  Reconsider \Cref{fig:abmc} and assume that we modify $\tau$ by conjoining $\ell = 0$, and \ref{accel1} by conjoining $\ell = 1$.
  Thus, we now have
  \begin{align*}
    \tau_{x<100} & {} \equiv x < 100 \land x' = x + 1 \land y' = y \land \ell = 0 & \text{and} \\
    \tau_{\inner}^+ & {} \equiv n > 0 \land x + n \leq 100 \land x' = x + n \land y' = y \land \ell = 1.
  \end{align*}
  When $b=2$, the trace is $[\tau_{x<100},\tau_{x<100}]$, and in the next iteration, it may be extended to either $\vec{\pi} = [\tau_{x<100},\tau_{x<100},\tau_{x<100}]$ or $\vec{\tau} = [\tau_{x<100},\tau_{x<100},\tau_{\inner}^+]$.
  However, as ${\to_{\tau_{\inner}^+}} = {\to^+_{\tau_{x < 100}}}$, we have ${\to_{\vec{\pi}}} \subseteq {\to_{\vec{\tau}}}$, so the entire search space can be covered without considering the trace $\vec{\pi}$.
  Thus, we add the blocking clause
  \[
    \neg\mu_2(\tau_{x < 100}) \tag{$\beta_1$}
  \]
  to the SMT problem to prevent ABMC from finding a model that gives rise to the trace $\vec{\pi}$.
  Note that we have $\mu_2(\tau_{\inner}^+) \models \beta_1$, as $\tau_{x<100} \models \ell = 0$ and $\tau_{\inner}^+ \models \ell \neq 0$.
  Thus, $\beta_1$ blocks $\tau_{x < 100}$ for the third step, but $\tau_{\inner}^+$ can still be used without restrictions.
  Therefore, adding $\beta_1$ to the SMT problem does not prevent us from covering the entire search space.

  Similarly, we have ${\to_{\vec{\pi}'}} \subseteq {\to_{\vec{\tau}}}$ for $\vec{\pi}' = [\tau_{x<100},\tau_{x<100},\tau_{\inner}^+,\tau_{x<100}]$.
  Thus, we also add the following blocking clause to the SMT problem:
  \[
    \ell^{(2)} \neq 1 \lor \neg\mu_3(\tau_{x < 100}) \tag{$\beta_2$}
  \]
\end{example}
\begin{algorithm}[t]
  \caption{$\ABMC_\block$ -- Input: a safety problem $\TT = (\psi_\init,\tau,\psi_\err)$}\label{alg:abmc-blocking}
  $b \gets 0; \ V \gets \emptyset; \ E \gets \emptyset; \ \id \gets 0; \ \tau \gets \tau \land \ell = 0; \cache \gets \emptyset; \ \add(\mu_b(\psi_\init))$ \label{abmcb:init}\;
  \leIf{$\neg\checksat()$}{
    \Return{$\safe$}
  }{
    $\sigma \gets \getmodel()$
  }
  \While{$\top$}{
     $\push()$; \quad $\add(\mu_b(\psi_\err))$\;
    \leIf{$\checksat()$\label{abmcb:err}}{
      \Return{$\unsafe$}
    }{
      $\pop()$
    }
    $\vec{\tau} \gets \trace_b(\sigma); \quad V \gets V \cup \vec{\tau}; \quad E \gets E \cup \{(\tau_1,\tau_2) \mid [\tau_1,\tau_2] \text{ is an infix of } \vec{\tau}\}$\;
    \If{$\vec{\tau} = \vec{\pi}\concat\vec{\pi}^\circlearrowleft \land \vec{\pi}^\circlearrowleft \text{ is $(V,E)$-cyclic} \land \shouldaccel(\vec{\pi}^\circlearrowleft)$ \label{abmcb:accel-cond}}{
      \If{$\exists \tau_c.\ (\vec{\pi}^\circlearrowleft, \tau_c) \in \cache$}{$\tau_\ell \gets \tau_c$ \label{abmcb:cache}\tcp*{the result of accelerating $\vec{\pi}^\circlearrowleft$ was cached}}
      \Else{
        $\id \gets \id + 1$; \ $\tau_\ell \gets \accel(\vec{\pi}^\circlearrowleft) \land \ell = \id$\tcp*{generate new ID and accelerate} \label{abmcb:accel}
        $\cache \gets \cache \cup \{(\vec{\pi}^\circlearrowleft, \tau_\ell)\}$\tcp*{update cache} \label{abmcb:learn}
      }
      $\beta_1 \gets \neg\left(\bigwedge_{i=0}^{|\vec{\pi}^\circlearrowleft|-1} \mu_{b+i}(\pi_i^\circlearrowleft) \right)$\tcp*{neither unroll $\vec{\tau}^\circlearrowleft$ right now...}\label{abmcb:block1}
      $\beta_2 \gets \ell^{(b)} \neq \id \lor \neg\left(\bigwedge_{i=0}^{|\vec{\pi}^\circlearrowleft|-1} \mu_{b+i+1}(\pi_i^\circlearrowleft)\right)$ \label{abmcb:block2}\tcp*{...nor after using the}
      $\add(\mu_{b}(\tau \lor \tau_\ell) \land \beta_1 \land \beta_2)$ \label{abmcb:add}\hspace{7em}\tcp*{accelerated transition}
    }
    \lElse{$\add(\mu_b(\tau))$  \label{abmcb:add-alternative}}
    \leIf{$\neg\checksat()$ \label{abmcb:sat-check}}{
      \Return{$\safe$}
    }{
      $\sigma \gets \getmodel(); \quad b \gets b + 1$ \label{abmcb:inc}
    }
  }
\end{algorithm}
ABMC with blocking clauses can be seen in \Cref{alg:abmc-blocking}.
The counter $\id$ is used to obtain unique identifiers for learned transitions.
Thus, it is initialized with $0$ (Line~\ref{abmcb:init}) and incremented whenever
a new transition is learned (Line~\ref{abmcb:accel}).
Moreover, as explained above, $\ell = 0$ is conjoined to $\tau$ (Line~\ref{abmcb:init}), and $\ell = \id$ is conjoined to each learned transition (Line~\ref{abmcb:accel}).

In Lines \ref{abmcb:block1} and \ref{abmcb:block2}, the blocking clauses $\beta_1$ and $\beta_2$ which correspond to
the superfluous traces $\vec{\pi}$ and $\vec{\pi}'$ above are created, and they are added
to the SMT problem in Line~\ref{abmcb:add}. Here, $\pi_i^\circlearrowleft$ denotes the
$i^{th}$ transition in the sequence $\vec{\pi}^\circlearrowleft$.

Importantly, \Cref{alg:abmc-blocking} caches (Line~\ref{abmcb:learn}) and reuses (Line~\ref{abmcb:cache}) learned transitions.
In this way, the learned transitions that are conjoined to the SMT problem have the same
$\id$ if they stem from the same cycle, and thus the blocking clauses $\beta_1$ and
$\beta_2$ can also block
sequences $\vec{\pi}^\circlearrowleft$
that contain learned transitions, as shown in the following example.

\begin{example}[Caching]\label{ex:caching}
  Let $\tau$ have the dependency graph given below. As \Cref{alg:abmc-blocking}
  \begin{minipage}{0.4\textwidth}
    \begin{tikzpicture}[bend angle=15]
      \node[draw=black] (A) at (0,0) {$\tau_1$};
      \node[draw=black] (B) at (2,0) {$\tau_2$};
      \node[draw=black] (C) at (4,0) {$\tau_3$};

      \path [->] (A) edge node[left] {} (B);
      \path [->] (B) edge[loop above] node[left] {} (B);
      \path [->] (B.north east) edge [bend left] node[left] {} (C.north west);
      \path [->] (C.south west) edge [bend left] node[left] {} (B.south east);
      \path [->] (C.north) edge [bend right=22] node[left] {} (A.north east);
    \end{tikzpicture}
  \end{minipage}
  \begin{minipage}{0.59\textwidth}
 conjoins $\ell = 0$ to $\tau$, assume $\tau_i \models \ell = 0$ for all $i \in \{1,2,3\}$.
  Moreover, assume that accelerating $\tau_2$ yields $\tau_2^+$ with $\tau_2^+ \models
  \ell = 1$. If we get the trace
  \end{minipage}

  \vspace{-0.1em}
  \noindent
  $[\tau_1,\tau_2^+,\tau_3]$, it can be accelerated.
  Thus, \Cref{alg:abmc-blocking} would add
  \[
    \beta_1 \equiv \neg\left(\mu_3(\tau_1) \land \mu_4(\tau^+_2) \land \mu_5(\tau_3)\right)
  \]
  to the SMT problem.
  If the next step yields the trace $[\tau_1,\tau_2^+,\tau_3,\tau_2]$, then $\tau_2$ is accelerated again.
  Without caching, acceleration may yield a new transition $\tau_{2'}^+$ with $\tau_{2'}^+ \models \ell = 2$.
  As the SMT solver may choose a different model in every iteration, the trace may also change in every iteration.
  So after two more steps, we could get the trace
  $[\tau_1,\tau_2^+,\tau_3,\tau_1,\tau_{2'}^+,\tau_3]$.
  At this point, the ``outer'' loop consisting of $\tau_1$, arbitrarily many repetitions of $\tau_2$, and $\tau_3$, has been unrolled a second time, which should have been prevented by $\beta_1$.
  The reason is that $\tau_{2}^+ \models \ell = 1$, whereas $\tau_{2'}^+ \models \ell = 2$, and thus $\tau_{2'}^+ \models \neg\tau^+_2$.
  With caching, we again obtain $\tau_2^+$ when $\tau_2$ is accelerated for the second time, such that this problem is avoided.
\end{example}
Remarkably, blocking clauses allow us to prove safety in cases where BMC fails.
\begin{example}[Proving Safety with Blocking Clauses]\label{ex:Proving Safety with Blocking Clauses}
  Consider the safety problem $(x \leq 0,\tau,x>100)$ with $\tau \equiv x < 100 \land x' = x + 1$.
  \Cref{alg:bmc} cannot prove its safety, as $\tau$ can
  be unrolled arbitrarily often (by choosing smaller and smaller initial values for $x$).
  With \Cref{alg:abmc-blocking}, we obtain the following SMT problem with $b=3$.
  \begin{align*}
    & \phantom{\neg}\mu_0(x \leq 0) \tag{initial states}\\
    & \phantom{\neg}\mu_0(\tau \land \ell = 0)\tag{$\tau$} \\
 & \phantom{\neg}\mu_1(\tau \land \ell = 0) \tag{$\tau$} \\
    & \neg\mu_2(\tau \land \ell = 0) \tag{$\beta_1$} \\
    & \phantom{\neg}\ell^{(2)} \neq 1 \lor \neg\mu_3(\tau \land \ell = 0) \tag{$\beta_2$} \\
    & \phantom{\neg}\mu_2((\tau \land \ell = 0) \lor (n > 0 \land x + n \leq 100 \land x' = x + n \land \ell = 1)) \tag{$\tau \lor \accel(\tau)$} \label{eq:ex-safe-accel} \\
    & \phantom{\neg}\mu_3(\tau \land \ell = 0) \tag{$\tau$}
  \end{align*}
  From the last formula and $\beta_2$, we get $\ell^{(2)} \neq 1$, but
 the formula labeled with \eqref{eq:ex-safe-accel} and $\beta_1$ imply $\mu_2(\ell = 1) \equiv \ell^{(2)} = 1$, resulting in a contradiction.
  Thus, due to the blocking clauses,
  $\ABMC_\block$ can prove safety with the bound $b=3$.
\end{example}

Like $\ABMC$, $\ABMC_\block$ preserves $\BMC$'s main properties
(see \cite{report} for a proof).
\begin{restatable}
  {theorem}{properties}
  \label{thm:properties}
  $\ABMC_\block$ is sound and refutationally complete, but non-terminating.
\end{restatable}
\makeproof*{thm:properties}{
  \properties*
  \begin{proof}
    Soundness for unsafety and non-termination are trivial.
    In the following, we prove refutational completeness.
    More precisely, we prove that when
    assuming a complete SMT solver, unsafety of $\TT$ implies that \Cref{alg:abmc-blocking} eventually returns $\unsafe$, independently of the non-determinisms of \Cref{alg:abmc-blocking} (i.e., independently of the models that are found by the SMT solver and the implementation of $\shouldaccel$).
    Together with soundness for unsafety, this means that \Cref{alg:abmc-blocking} returns
    $\unsafe$ if and only if $\TT$ is unsafe, which in turn implies soundness for
    safety.

    The proof proceeds as follows:
    We consider an arbitrary but fixed unsafe safety problem $(\psi_\init,\tau,\psi_\err)$.
Thus, there is a model $\sigma'$ and a sequence  $[\tau_i]_{i=0}^{c-1} \in \sip(\tau)^c$ of
original transitions such that $\sigma' \models \mu_0(\psi_\init)$, $\sigma'(\vec{x}^{(i)})
\to_{\tau_i} \sigma'(\vec{x}^{(i+1)})$ for all $0 \leq i < c$, and $\sigma' \models
\mu_c(\psi_\err)$.
So in particular, $\sigma'$ can be extended to a substitution $\sigma$ such
that
 \begin{equation}
      \label{eq:path-sat}
      \sigma \models \mu_0(\psi_\init) \land \bigwedge_{i=0}^{c-1} \mu_i(\tau_i),
    \end{equation}
 i.e.,
 $[\tau_i]_{i=0}^{c-1}$ and $\sigma$ witness that the state $\sigma(\vec{x}^{(c)})$ is reachable from the initial state $\sigma(\vec{x}^{(0)})$.
 We now have to prove that every run of \Cref{alg:abmc-blocking} yields $\unsafe$.

   To this end, we define a class of infinite sequences $[\phi_j]_{j \in \NN}$ of formulas with a certain structure.
    This class is defined in such a way that every run of \Cref{alg:abmc-blocking} corresponds to a prefix of such an infinite sequence in the sense that $\phi_j$ is the formula that is added to the SMT problem in Line~\ref{abmcb:add} or \ref{abmcb:add-alternative} when $b = j$.
    Note that the converse is not true, i.e., there are sequences of formulas $[\phi_j]_{j \in \NN}$ with the desired structure that do not correspond to a run of \Cref{alg:abmc-blocking}.
    Next, we choose such a sequence $[\phi_j]_{j \in \NN}$ 
    arbitrarily.
    Let
    \[
      \xi_m \Def \mu_0(\psi_\init) \land \bigwedge_{j=0}^{m} \phi_j \text{ for any $m \in
        \NN$ \quad and \quad} \xi_{-1} \Def \mu_0(\psi_\init).
          \]
          Then our goal is to prove that there is a $k \in \NN$ and a model $\theta$ of $\xi_{k-1}$
              such that $\theta(\vec{x}^{(0)}) = \sigma(\vec{x}^{(0)})$ and $\theta(\vec{x}^{(k)}) = \sigma(\vec{x}^{(c)})$.
    Note that if $[\phi_j]_{j \in \NN}$ corresponds to a run of \Cref{alg:abmc-blocking},
    then $\xi_{k-1}$ is the formula that is checked by \Cref{alg:abmc-blocking} in
    Line~\ref{abmcb:sat-check} when $b = k-1$.
    Now assume that $\sigma(\vec{x}^{(c)})$ is an error state, i.e., $\sigma \models \mu_c(\psi_\err)$.
    Then \Cref{alg:abmc-blocking} checks the formula $\xi_{k-1} \land \mu_k(\psi_\err)$ in
    Line~\ref{abmcb:err} in the next iteration (when $b=k$).
    As $\theta \models \xi_{k-1}$, $\sigma \models \mu_c(\psi_\err)$, and $\theta(\vec{x}^{(k)}) = \sigma(\vec{x}^{(c)})$, we get $\theta \models \xi_{k-1} \land \mu_k(\psi_\err)$.
    Thus, \Cref{alg:abmc-blocking} returns $\unsafe$ when $b = k$.
    As $[\phi_j]_{j \in \NN}$ was chosen arbitrarily, this implies refutational completeness.

    In the sequel, we introduce various notations, which are summarized in
    \Cref{notations} for easy reference. We now define our class of infinite sequences
    $[\phi_j]_{j \in \NN}$.
    We have
    \begin{figure}[th!]
      \fbox{\begin{minipage}{0.98\textwidth}
    \begin{description}
    \item[$\rho$:] an original or learned transition
    \item[\protect{$[\phi_j]_{j \in \NN}$:}] an infinite sequence of formulas such that
every $\ABMC_\block$-run corresponds to a prefix of such a sequence, as explained above.
      By choosing this sequence arbitrarily, we cover all possible $\ABMC_\block$-runs.
    \item[$\chi_j$:] a formula such that $\mu_j(\chi_j)$ is equivalent to $\phi_j$ without blocking clauses
    \item[$\pi_{j}$:] the transition that is learned by \Cref{alg:abmc-blocking} when $b = j$, if any
    \item[$\id(\rho)$:] We define $\id(\rho) = 0$ if $\rho \in \sip(\tau)$, i.e., if $\rho$ is an original transition.
      Otherwise, $\id(\rho)$ is a unique identifier of $\rho$.
    \item[$\vec{\rho}^j$:] a vector of transitions such that ${\to_{\rho}} \subseteq {\to^+_\tau}$ for each $\rho \in \vec{\rho}^j$.
      Intuitively, $\vec{\rho}^j$ corresponds to the cyclic suffix $\vec{\pi}^\circlearrowleft$ when $b = j$, if any.
      More precisely, then we have $\vec{\pi}^\circlearrowleft = [\rho^j_i \land \ell =
        \id(\rho^j_i)]_{i=0}^{|\vec{\rho}^j|-1}$.
    \item[$\beta_j$:] the conjunction of all
blocking clauses that are added by \Cref{alg:abmc-blocking} when $b = j$
      \begin{description}
      \item[$\beta_{1,j},\beta_{2,j}$:] the blocking clauses that are constructed in Line~\ref{abmcb:block2} when $b=j$, if any
      \end{description}
    \item[\protect{$[\tau_i]_{i=0}^{c-1}$}:] a sequence of original transitions
    \item[$\sigma$:] a model of $\mu_0(\psi_\init)$ such that $\sigma(\vec{x}^{(0)}) \to_{\tau_0} \sigma(\vec{x}^{(1)}) \to_{\tau_1} \ldots \to_{\tau_{c-1}} \sigma(\vec{x}^{(c)})$, i.e., $\sigma$ witnesses reachability of $\sigma(\vec{x}^{(1)}),\ldots,\sigma(\vec{x}^{(c)})$ (from an initial state)
    \item[$k$:] a number for which we will prove that \Cref{alg:abmc-blocking} shows reachability of $\sigma(\vec{x}^{(c)})$ when $b=k-1$, at the latest
    \item[$c_j$:] for each $j \in \{0,\ldots,k\}$, $c_j$ is a number such that \Cref{alg:abmc-blocking} shows reachability of $\sigma(\vec{x}^{(c_j)})$ when $b=j - 1$, at the latest
    \item[$\xi_m$:] as defined above -- the formula that the $\ABMC_\block$-run which
      corresponds to $[\phi_j]_{j=0}^{k-1}$ (if any) checks in Line~\ref{abmcb:sat-check} when
      $b=m$. So $\xi_m$ may contain $\vec{x}^{(0)}, \ldots, \vec{x}^{(m+1)}$.
    \item[$\theta_j$:] for each $j \in \{0,\ldots,k\}$, $\theta_j$ is a partial model of
      $\xi_{j-1}$, i.e., $\theta_j(\xi_{j-1})$ is satisfiable. So the domain of
      $\theta_j$ may contain $\vec{x}^{(0)}, \ldots, \vec{x}^{(j)}$.
    \item[$\eta_j$:] for each $j \in \{0,\ldots,k\}$, $\eta_j$ instantiates the variables
      that occur in $\pi_j$, but not in $\vec{\rho}^j$, i.e., the additional variables
      introduced by acceleration
    \item[$\theta$:] an extension of $\theta_k$ such that $\theta \models \xi_{k-1}$
    \item[$\mu_j^{\pre}$:] a substitution such that $\mu_j^{\pre}(\vec{x}') = \vec{x}'$ and $\mu_j^{\pre}(x) = x^{(j)}$ if $x \notin \vec{x}'$.
      In other words, $\mu_j^{\pre}$ behaves like $\mu_j$ on all but the post-variables, where it is the identity.
    \end{description}
    \end{minipage}}
    \caption{Notation used throughout the proof}\label{notations}
    \end{figure}
     \[
      \phi_j \Def \mu_j(\chi_j) \land \beta_j \qquad \text{and} \qquad \chi_j \Def (\tau \land \ell = 0) \lor \pi_{j}
    \]
    where one of the following holds for each $j \in \NN$:
    \begin{align*}
      \pi_{j} \equiv \bot \qquad \text{and} \qquad \beta_j \equiv \top
    \end{align*}
    or
    \begin{align*}
      \pi_{j} & {} \equiv \accel(\vec{\rho}^{j}) \land \ell = \id(\accel(\vec{\rho}^{j})), \\
      \beta_{1,j} & {} \equiv  \neg\left(\bigwedge_{i=0}^{|\vec{\rho}^{j}|-1} \mu_{j+i}(\rho_i^{j} \land \ell = \id(\rho_i^{j}))\right), \\
      \beta_{2,j} & {} \equiv \ell^{(j)} \neq \id(\accel(\vec{\rho}^{j})) \lor \neg\left(\bigwedge_{i=0}^{|\vec{\rho}^{j}|-1} \mu_{j+i+1}(\rho_i^{j} \land \ell = \id(\rho_i^{j}))\right), \qquad \text{and} \\
      \beta_j & {} \equiv \beta_{1,j} \land \beta_{2,j}
    \end{align*}
    Here, the literals $\ell = \id(\rho^j_i)$ are contained in $\beta_{1,j}$ and $\beta_{2,j}$ as \Cref{alg:abmc-blocking} conjoins $\ell = 0$ to $\tau$ and $\ell = \id$ to learned transitions.
    So we assume that $\rho^j_i$ is an (original or learned) transition without the additional literal for $\ell$, and we made those literals explicit in the definitions of $\beta_{1,j}$ and $\beta_{2,j}$ for clarity.
    Thus,  $\phi_j$ is equivalent to the formula that is added to the SMT problem in Line~\ref{abmcb:add-alternative} or \ref{abmcb:add}.

 To obtain $k$ and $\theta$, we now define sequences $[c_j]_{j=0}^k$ and
    $[\theta_j]_{j=0}^k$ inductively
    such that $\theta_j(\xi_{j-1})$ is satisfiable
    and $\theta_j(x^{(j)}) = \sigma(x^{(c_j)})$ for all $x \neq \ell$ such that $x^{(c_j)} \in \dom(\sigma)$. Here, two problems
    have to be solved: The first problem is that \Cref{alg:abmc-blocking} may replace many
    evaluation steps by a single step using an accelerated transition. Therefore, the
    $j^{th}$ step in the algorithm does not necessarily correspond to the $j^{th}$ step in the
    evaluation with the original transitions $[\tau_i]_{i=0}^{c-1}$, but to the $c^{th}_j$ step where $c_j \geq
    j$. The second problem is that the accelerated transitions may contain additional
    variables which have to be instantiated appropriately (to this end, we use
    substitutions $\eta_j$).
    \begin{enumerate}[label=\alph*.]
    \item \label{it:step-init} $c_0 \Def 0$ and $\theta_0 \Def [x^{(0)}/\sigma(x^{(0)})
      \mid x \notin \vec{x}']$
    \item \label{it:step-term} if $c_j = c$, then $k \Def j$, i.e., then the definition of $[c_j]_{j=0}^k$ and $[\theta_{j}]_{j=0}^{k}$ is complete
    \item otherwise:
      \begin{enumerate}[label*=\alph*.]
      \item \label{it:step1} if $\pi_{j} \equiv \bot$, then:
        \begin{enumerate}[label*=\alph*.]
        \item \label{it:step1-c} $c_{j+1} \Def c_j + 1$
        \item \label{it:step1-sigma} $\theta_{j+1}(\ell^{(j)}) \Def 0$
        \end{enumerate}
      \item \label{it:f} if $\pi_{j} \equiv \accel(\vec{\rho}^{j}) \land \ell = \id(\accel(\vec{\rho}^{j}))$ let $f_j$ be the maximal natural number such that $\sigma(\vec{x}^{(c_j)}) \to_{\pi_{j}} \sigma(\vec{x}^{(c_j + f_j)})$, or $0$ if $\sigma(\vec{x}^{(c_j)}) \not\to_{\pi_{j}} \sigma(\vec{x}^{(c_j + f_j)})$ for all $f_j > 0$
        \begin{enumerate}[label*=\alph*.]
        \item \label{it:step2} if $f_j = 0$, then:
          \begin{enumerate}[label*=\alph*.]
          \item \label{it:step2-c} $c_{j+1} \Def c_j+1$
          \item \label{it:step2-sigma} $\theta_{j+1}(\ell^{(j)}) \Def 0$
          \end{enumerate}
        \item \label{it:step3} if $f_j>0$, then:
          \begin{enumerate}[label*=\alph*.]
          \item \label{it:step3-c} $c_{j+1} \Def c_{j}+f_j$
          \item \label{it:step3-eta} $\theta_{j+1}(x^{(j)}) \Def \eta_j(x^{(c_{j})})$ for all $x^{(c_{j})} \in \VV(\mu_{c_j}(\pi_j)) \setminus \dom(\sigma)$, where $\eta_j$ is a substitution with $\eta_j \circ \sigma \circ (\mu_{c_j}^{\pre} \uplus [\vec{x}' / \vec{x}^{(c_j + f_j)}]) \models \pi_{j}$ (which exists due to $\sigma(\vec{x}^{(c_j)}) \to_{\pi_{j}} \sigma(\vec{x}^{(c_j + f_j)})$)
          \end{enumerate}
        \end{enumerate}
      \item \label{it:step-subs1} $\theta_{j+1}(x^{(m)}) \Def \theta_{j}(x^{(m)})$ for all $x^{(m)} \in \dom(\theta_j)$
      \item \label{it:step-subs2} $\theta_{j+1}(x^{(j+1)}) \Def \sigma(x^{(c_{j+1})})$ for all $x^{(c_{j+1})} \in \dom(\sigma)$
      \end{enumerate}
    \end{enumerate}
    Note that the steps \ref{it:step3-eta} and \ref{it:step-subs1} do not contradict each other.
    This is obvious, except for the case that $m=j$ in Step \ref{it:step-subs1}\
    Then $x^{(j)} \in \dom(\theta_j)$.
    Thus, $[x^{(j)}/\theta_j(x^{(j)})]$ has been added to $\theta_j$ in Step \ref{it:step-subs2} (in all other steps where entries are added to $\theta_j$, the indices of the variable(s) and the substitution differ).
    Thus, we have $x^{(c_j)} \in \dom(\sigma)$, and Step \ref{it:step3-eta} explicitly excludes elements of $\dom(\sigma)$.
    Moreover, by construction, we have $c_m < c_{m+1}$, and thus the sequences are well defined.

    Let $\theta(x) \Def \theta_k(x)$ if $x \in \dom(\theta_k)$, and for all $x \notin \dom(\theta_k)$,
    let $\theta(x)$ be a natural number which is larger than all constants that
    occur in $\xi_{k-1}$.
    To see why it is necessary to instantiate variables that are not contained in $\dom(\theta_k)$, recall that $\xi_{k-1}$ contains blocking clauses, and these blocking clauses may refer to ``future steps'', i.e., to steps that are not yet part of the SMT encoding.
 For example, it may happen that $\xi_{k-1}$ contains a blocking clause that refers to
 $x^{(m)}$ with $m > k$.
Similarly, 
$\xi_{k-1}$ may contain a blocking clause referring to $\ell^{(k)}$, which is also not contained in $\dom(\theta_k)$.
    The reason is that $\ell^{(j)}$'s value is the identifier of the transition that is used for the step from $\sigma(\vec{x}^{(c_j)})$ to $\sigma(\vec{x}^{(c_{j+1})})$, but $\sigma(\vec{x}^{(c_{k})}) = \sigma(\vec{x}^{(c)})$ is the last state of the run $\sigma(\vec{x}^{(0)}) \to_\tau \ldots \to_\tau \sigma(\vec{x}^{(c)})$ that we are considering.

    We prove $\theta \models \xi_{k-1}$.
    To this end, we prove the following three statements individually:
    \begin{align}
      \label{goal1}
      \theta & {} \models \mu_0(\psi_\init)\\
      \label{goal2}
      \theta & {} \models \bigwedge_{j=0}^{k-1} \mu_j(\chi_j) \\
      \label{goal3}
      \theta & {} \models \bigwedge_{j=0}^{k-1} \beta_j
    \end{align}
    Then $\theta \models \xi_{k-1}$ follows due to the definition of
    $\xi_{k-1}$.

    For \eqref{goal1}, we have
    \begin{align*}
      &&\sigma & {} \models \mu_0(\psi_\init) \tag{\ref{eq:path-sat}} \\
      \curvearrowright && \theta_{0} & {} \models \mu_0(\psi_\init) \tag{\ref{it:step-init}} \\
      \curvearrowright && \theta_{k} & {} \models \mu_0(\psi_\init) \tag{\ref{it:step-subs1}} \\
      \curvearrowright && \theta & {} \models \mu_0(\psi_\init) \tag{def.\ of $\theta$} \\
    \end{align*}

    For \eqref{goal2}, let $j \in \{0,\ldots,k-1\}$ be arbitrary but fixed.
    We prove $\theta \models \mu_j(\chi_j)$.
    First consider the case $\pi_{j} \equiv \bot$, see \eqref{it:step1}.
    Then:
    \begin{align*}
      && \sigma & {} \models \mu_{c_j}(\tau_{c_j}) \tag{\eqref{eq:path-sat}, as $j < k$ implies $c_j < c$} \\
      \curvearrowright && \sigma & {} \models \mu_{c_j}(\tau) \tag{$\tau_{c_j} \in \sip(\tau)$} \\
      \curvearrowright && \sigma \circ \mu_{c_j} & {} \models \tau \\
      \curvearrowright && \sigma \circ (\mu_{c_j}^{\pre} \uplus [\vec{x}'/\vec{x}^{(c_j+1)}]) & {} \models \tau \tag{def.\ of $\mu_{c_j}$ and $\mu_{c_j}^{\pre}$}\\
      \curvearrowright && (\sigma \circ \mu_{c_j}^{\pre}) \uplus [\vec{x}'/\sigma(\vec{x}^{(c_j+1)})] & {} \models \tau \\
      \curvearrowright && (\sigma \circ \mu_{c_j}^{\pre}) \uplus [\vec{x}'/\sigma(\vec{x}^{(c_{j+1})})] & {} \models \tau \tag{\ref{it:step1-c}}\\
      \curvearrowright && (\theta_{j} \circ \mu_{j}^{\pre}) \uplus [\vec{x}'/\theta_{j+1}(\vec{x}^{(j+1)})] & {} \models \tau \tag{\ref{it:step-subs2}} \\
      \curvearrowright && (\theta_{j+1} \circ \mu_{j}^{\pre}) \uplus [\vec{x}'/\theta_{j+1}(\vec{x}^{(j+1)})] & {} \models \tau \tag{\ref{it:step-subs1}} \\
      \curvearrowright && \theta_{j+1} \circ (\mu_{j}^{\pre} \uplus [\vec{x}'/\vec{x}^{(j+1)}]) & {} \models \tau \\
      \curvearrowright && \theta_{j+1} \circ \mu_{j} & {} \models \tau \tag{def.\ of $\mu_{j}$ and $\mu_{j}^{\pre}$} \\
      \curvearrowright && \theta_{j+1} & {} \models \mu_{j}(\tau) \\
      \curvearrowright && \theta_{j+1} \circ [\ell^{(j)}/0] & {} \models \mu_{j}(\chi_{j}) \tag{$\tau \land \ell = 0 \models \chi_{j}$} \\
      \curvearrowright && \theta_{j+1} & {} \models \mu_{j}(\chi_{j}) \tag{\ref{it:step1-sigma}} \\
      \curvearrowright && \theta_{k} & {} \models \mu_{j}(\chi_{j}) \tag{\ref{it:step-subs1}} \\
      \curvearrowright && \theta & {} \models \mu_{j}(\chi_{j}) \tag{def.\ of $\theta$}
    \end{align*}
    The case \eqref{it:step2} is analogous.
    Now consider the case that $\pi_{j} \not\equiv \bot$ and $f_j > 0$, see \eqref{it:step3}.
    Then:
    \begin{align*}
      && \eta_j \circ \sigma \circ (\mu_{c_j}^{\pre} \uplus [\vec{x}' / \vec{x}^{(c_j + f_j)}]) & {} \models \pi_{j} \tag{\ref{it:step3-eta}} \\
      \curvearrowright && \eta_j \circ \sigma \circ (\mu_{c_j}^{\pre} \uplus [\vec{x}' / \vec{x}^{(c_{j+1})}]) & {} \models \pi_{j} \tag{\ref{it:step3-c}} \\
      \curvearrowright && (\eta_j \circ \sigma \circ \mu_{c_j}^{\pre}) \uplus [\vec{x}' / \sigma(\vec{x}^{(c_{j+1})})] & {} \models \pi_{j} \tag{$\dom(\eta_j) \cap \vec{x}^{(c_{j+1})} = \emptyset$} \\
      \curvearrowright && ((\eta_j \circ \mu_{c_j}^{\pre}) \circ (\sigma \circ
      \mu_{c_j}^{\pre})) \uplus [\vec{x}' / \sigma(\vec{x}^{(c_{j+1})})] & {} \models
      \pi_{j} \tag{$\sigma$ is a ground substitution}\\
      \curvearrowright && ((\eta_j \circ \mu_{c_j}^{\pre}) \circ (\theta_j \circ \mu_{j}^{\pre})) \uplus [\vec{x}' / \theta_{j+1}(\vec{x}^{(j+1)})] & {} \models \pi_{j} \tag{\ref{it:step-subs2}} \\
      \curvearrowright && ((\eta_j \circ \mu_{c_j}^{\pre}) \circ (\theta_{j+1} \circ \mu_{j}^{\pre})) \uplus [\vec{x}' / \theta_{j+1}(\vec{x}^{(j+1)})] & {} \models \pi_{j} \tag{\ref{it:step-subs1}} \\
      \curvearrowright && ((\theta_{j+1} \circ \mu_{j}^{\pre}) \circ (\theta_{j+1} \circ \mu_{j}^{\pre})) \uplus [\vec{x}' / \theta_{j+1}(\vec{x}^{(j+1)})] & {} \models \pi_{j} \tag{\ref{it:step3-eta}} \\
      \curvearrowright && (\theta_{j+1} \circ \mu_{j}^{\pre}) \uplus [\vec{x}' /
        \theta_{j+1}(\vec{x}^{(j+1)})] & {} \models \pi_{j} \tag{$\theta_{j+1}$ is a ground substitution} \\
      \curvearrowright && \theta_{j+1} \circ (\mu_{j}^{\pre} \uplus [\vec{x}' / \vec{x}^{(j+1)}]) & {} \models \pi_{j} \\
      \curvearrowright && \theta_{j+1} \circ \mu_{j} & {} \models \pi_{j} \tag{def.\ of $\mu_{j}^{\pre}$ and $\mu_{j}$} \\
      \curvearrowright && \theta_{j+1} & {} \models \mu_{j}(\pi_{j}) \\
      \curvearrowright && \theta_{j+1} 
      & {} \models \mu_j(\chi_{j}) \tag{$\pi_j  \models \chi_{j}$} \\
      \curvearrowright && \theta_{k} & {} \models \mu_{j}(\chi_{j}) \tag{\ref{it:step-subs1}} \\
      \curvearrowright && \theta & {} \models \mu_{j}(\chi_{j}) \tag{def.\ of $\theta$} \\
    \end{align*}
    This finishes the proof of \eqref{goal2}.

    For \eqref{goal3}, let $j \in \{0,\ldots,k-1\}$ be arbitrary but fixed.
    We prove $\theta \models \beta_j$.
    If $\beta_j \equiv \top$, the claim is trivial, so assume $\beta_j \not\equiv \top$, i.e., $\beta_j \equiv \beta_{1,j} \land \beta_{2,j}$.
    We prove $\theta \models \beta_{1,j}$ and $\theta \models \beta_{2,j}$ individually.

    To prove $\theta \models \beta_{1,j}$, we have to show
    \begin{equation}
      \label{goal4}
      \theta \models \neg\left(\bigwedge_{i=0}^{|\vec{\rho}^{j}|-1} \mu_{j+i}(\rho_i^{j} \land \ell = \id(\rho_i^{j}))\right)
    \end{equation}
    by definition of $\beta_{1,j}$.
    First consider the case $j + |\vec{\rho}^{j}| > k$.
    Then:
    \begin{align*}
      &&& \phantom{{} \models {}} \ell^{(j + |\vec{\rho}^{j}| - 1)} \notin \dom(\theta_k) \\
      \curvearrowright && \theta & {} \models \ell^{(j+|\vec{\rho}^{j}|-1 )} \neq \id(\rho_{|\vec{\rho}^{j}|-1}^{j}) \tag{def.\ of $\theta$, as $\id(\rho_{|\vec{\rho}^{j}|-1}^{j})$ occurs in $\xi_{k-1}$} \\
      \curvearrowright && \theta & {} \models \mu_{j+|\vec{\rho}^{j}|-1}(\ell \neq \id(\rho_{|\vec{\rho}^{j}|-1}^{j})) \tag{def.\ of $\mu_{j+|\vec{\rho}^{j}|-1}$} \\
      \curvearrowright &&& \phantom{{} \models {}} \eqref{goal4}
    \end{align*}
    Now consider the case $j + |\vec{\rho}^{j}| \leq k$.
    First assume $f_j > 0$.
    Then:
    \begin{align*}
      && \eta_j \circ \sigma \circ (\mu_{c_j}^{\pre} \uplus [\vec{x}' / \vec{x}^{(c_j + f_j)}]) & {} \models \pi_{j} \tag{\ref{it:step3-eta}} \\
      \curvearrowright&& \eta_j \circ \sigma \circ (\mu_{c_j}^{\pre} \uplus [\vec{x}' / \vec{x}^{(c_j + f_j)}]) & {} \models \ell = \id(\accel(\vec{\rho}^{j})) \tag{$\pi_j \models \ell = \id(\accel(\vec{\rho}^{j}))$} \\
      \curvearrowright && \eta_j \circ \sigma & {} \models \ell^{(c_j)} = \id(\accel(\vec{\rho}^{j})) \tag{def.\ of $\mu_{c_j}^{\pre}$} \\
      \curvearrowright && \eta_j & {} \models \ell^{(c_j)} = \id(\accel(\vec{\rho}^{j})) \tag{$\ell^{(c_j)} \notin \dom(\sigma)$} \\
      \curvearrowright && \theta_{j+1} & {} \models \ell^{(j)} = \id(\accel(\vec{\rho}^{j})) \tag{\ref{it:step3-eta}} \\
      \curvearrowright && \theta_{k} & {} \models \ell^{(j)} = \id(\accel(\vec{\rho}^{j})) \tag{\ref{it:step-subs1}} \\
      \curvearrowright && \theta & {} \models \ell^{(j)} = \id(\accel(\vec{\rho}^{j})) \tag{def.\ of $\theta$} \\
      \curvearrowright && \theta & {} \models \ell^{(j)} \neq \id(\rho_{0}^{j}) \tag{$\id(\rho_{0}^{j}) \neq \id(\accel(\vec{\rho}^{j}))$} \\
      \curvearrowright && \theta & {} \models \neg \mu_j(\ell = \id(\rho_{0}^{j})) \tag{def.\ of $\mu_j$}\\
      \curvearrowright &&& \phantom{{} \models {}}  \eqref{goal4}
    \end{align*}
    Now assume $f_j = 0$.
    Then:
    \begin{align*}
      && \sigma(\vec{x}^{(c_j)}) & {} \not\to_{\pi_j} \sigma(\vec{x}^{(c_j + g)}) \tag{for all $g > 0$, as $f_j = 0$} \\
      \curvearrowright &&\sigma(\vec{x}^{(c_j)}) & {} \not\to_{\pi_j} \sigma(\vec{x}^{(c_{j + |\vec{\rho}^{j}|})}) \tag{$j+|\vec{\rho}^{j}| \leq k$ and $c_{j+|\vec{\rho}^{j}|} > c_j$} \\
      \curvearrowright &&\sigma(\vec{x}^{(c_j)}) & {} \not\to_{\vec{\rho}^j} \sigma(\vec{x}^{(c_{j + |\vec{\rho}^{j}|})}) \tag{${\to_{\pi_j}} = {\to_{\vec{\rho}^j}^+}$} \\
      \curvearrowright &&\theta_j(\vec{x}^{(j)}) & {} \not\to_{\vec{\rho}^j} \theta_{j + |\vec{\rho}^{j}|}(\vec{x}^{(j + |\vec{\rho}^{j}|)}) \tag{\ref{it:step-subs2}} \\
      \curvearrowright &&\theta_{j + |\vec{\rho}^{j}|}(\vec{x}^{(j)}) & {} \not\to_{\vec{\rho}^j} \theta_{j + |\vec{\rho}^{j}|}(\vec{x}^{(j + |\vec{\rho}^{j}|)}) \tag{\ref{it:step-subs1}} \\
      \curvearrowright && \theta_{j + |\vec{\rho}^{j}|} & {} \models \neg\left(\bigwedge_{i=0}^{|\vec{\rho}^{j}|-1} \mu_{j+i}(\rho_i^{j})\right) \tag{def.\ of $\to_{\vec{\rho}^{j}}$} \\
      \curvearrowright && \theta_{j + |\vec{\rho}^{j}|} & {} \models \neg\mu_{j+i}(\rho_i^{j}) \tag{for some $0 \leq i < |\vec{\rho}^{j}|$} \\
      \curvearrowright && \theta_{k} & {} \models \neg\mu_{j+i}(\rho_i^{j}) \tag{\ref{it:step-subs1}} \\
      \curvearrowright && \theta & {} \models \neg\mu_{j+i}(\rho_i^{j}) \tag{def.\ of $\theta$, as $j+i < j + |\vec{\rho}^{j}| \leq k$} \\
      \curvearrowright && & \phantom{{} \models {}} \eqref{goal4}
    \end{align*}
    This finishes the proof of \eqref{goal4}.

    To prove $\theta \models \beta_{2,j}$, we have to show
    \begin{equation}
      \label{goal5}
      \theta \models \ell^{(j)} \neq \id(\accel(\vec{\rho}^{j})) \lor \neg\left(\bigwedge_{i=0}^{|\vec{\rho}^{j}|-1} \mu_{j+i+1}(\rho_i^{j} \land \ell = \id(\rho_i^{j}))\right)
    \end{equation}
    by definition of $\beta_{2,j}$.
    First consider the case $j + |\vec{\rho}^{j}| \geq k$.
    Then:
    \begin{align*}
      &&& \phantom{{} \models {}} \ell^{(j + |\vec{\rho}^{j}|)} \notin \dom(\theta_k) \\
      \curvearrowright && \theta & {} \models \ell^{(j + |\vec{\rho}^{j}|)} \neq \id(\rho_{|\vec{\rho}^{j}|-1}^{j}) \tag{def.\ of $\theta$, as $\id(\rho_{|\vec{\rho}^{j}|-1}^{j})$ occurs in $\xi_{k-1}$} \\
      \curvearrowright && \theta & {} \models \neg\mu_{j+|\vec{\rho}^{j}|}(\ell = \id(\rho_{|\vec{\rho}^{j}|-1}^{j})) & \tag{def.\ of $\mu_{j+|\vec{\rho}^{j}|}$} \\
      \curvearrowright &&& \phantom{{} \models {}} \eqref{goal5}
    \end{align*}
    Now consider the case $j + |\vec{\rho}^{j}| < k$.
    First assume $f_j = 0$.
    Then:
    \begin{align*}
      && \theta_{j+1} & {} \models \ell^{(j)} \neq \id(\accel(\vec{\rho}^{j})) \tag{\ref{it:step2-sigma}} \\
      \curvearrowright && \theta_{k} & {} \models \ell^{(j)} \neq \id(\accel(\vec{\rho}^{j})) \tag{\ref{it:step-subs1}} \\
      \curvearrowright && \theta & {} \models \ell^{(j)} \neq \id(\accel(\vec{\rho}^{j})) \tag{def.\ of $\theta$} \\
      \curvearrowright &&& \phantom{{} \models {}} \eqref{goal5}
    \end{align*}
    Now assume $f_j > 0$.
    Then:
    \begin{align*}
      && \sigma(\vec{x}^{(c_j + f_j)}) & {} \not\to_{\pi_j} \sigma(\vec{x}^{(c_j + f_j + g)}) \tag{for all $g > 0$, as $f_j$ is maximal} \\
      \curvearrowright && \sigma(\vec{x}^{(c_{j+1})}) & {} \not\to_{\pi_j} \sigma(\vec{x}^{(c_{j+1} + g)}) \tag{\ref{it:step3-c}} \\
      \curvearrowright && \sigma(\vec{x}^{(c_{j+1})}) & {} \not\to_{\pi_j} \sigma(\vec{x}^{(c_{j + 1 + |\vec{\rho}^j|})}) \tag{$j + |\vec{\rho}^j| < k$ and $c_{j+1} < c_{j+1+|\vec{\rho}^j|}$} \\
      \curvearrowright && \sigma(\vec{x}^{(c_{j+1})}) & {} \not\to_{\vec{\rho}^j} \sigma(\vec{x}^{(c_{j + 1 + |\vec{\rho}^j|})}) \tag{${\to_{\pi_j}} = {\to_{\vec{\rho}^j}^+}$} \\
      \curvearrowright && \theta_{j+1}(\vec{x}^{(j+1)}) & {} \not\to_{\vec{\rho}^j} \theta_{j + 1 + |\vec{\rho}^j|}(\vec{x}^{(j + 1 + |\vec{\rho}^j|)}) \tag{\ref{it:step-subs2}} \\
      \curvearrowright && \theta_{j+1+|\vec{\rho}^j|}(\vec{x}^{(j+1)}) & {} \not\to_{\vec{\rho}^j} \theta_{j + 1 + |\vec{\rho}^j|}(\vec{x}^{(j + 1 + |\vec{\rho}^j|)}) \tag{\ref{it:step-subs1}} \\
      \curvearrowright && \theta_{j+1+|\vec{\rho}^j|} & {} \models \neg\left(\bigwedge_{i=0}^{|\vec{\rho}^{j}|-1} \mu_{j+1+i}(\rho_i^{j})\right) \tag{def.\ of $\to_{\vec{\rho}^j}$} \\
      \curvearrowright && \theta_{j+1+|\vec{\rho}^j|} & {} \models \neg\mu_{j+1+i}(\rho_i^{j}) \tag{for some $i \in \{0,\ldots,|\vec{\rho}^j|-1\}$} \\
      \curvearrowright && \theta_{k} & {} \models \neg\mu_{j+1+i}(\rho_i^{j}) \tag{\ref{it:step-subs1}} \\
      \curvearrowright && \theta & {} \models \neg\mu_{j+1+i}(\rho_i^{j}) \tag{def.\ of $\theta$, as $j+1+i < j+1+|\vec{\rho}^j| \leq k$} \\
      \curvearrowright &&& \phantom{{} \models {}} \eqref{goal5}
    \end{align*}
    This finishes the proof of \eqref{goal5}, and hence also the proof of \eqref{goal3}.
    Thus, we have proven $\theta \models \xi_{k-1}$.

    Now assume that $\sigma(\vec{x}^{(c)})$ is an error state, i.e., $\sigma \models
    \mu_c(\psi_\err)$.
    For all $x \neq \ell$ such that $x^{(c_k)} \in \dom(\sigma)$, we have
    \begin{align*}
      \theta(x^{(k)}) & {} = \theta_k(x^{(k)}) \tag{by def.\ of $\theta$} \\
                      & {} = \sigma(x^{(c_k)}) \tag{by \ref{it:step-subs2}} \\
                      & {} = \sigma(x^{(c)}) \tag{by \ref{it:step-term}}.
    \end{align*}
    Thus, we get $\theta \models \xi_{k-1} \land \mu_k(\psi_\err)$.
    Hence, the run of \Cref{alg:abmc-blocking} that checks the formula $\xi_{k-1}$ when
    $b=k-1$ returns $\unsafe$ when $b=k$.
    Note that satisfiability of $\xi_{k-1}$ implies satisfiability of all formulas that
    are checked by \Cref{alg:abmc-blocking} in Line~\ref{abmcb:sat-check} when $b <
    k-1$.
    Thus, \Cref{alg:abmc-blocking} cannot return $\safe$ when $b < k$.
    As $[\phi_j]_{j \in \NN}$ was chosen arbitrarily,
    this implies refutational completeness. \qed
  \end{proof}
}

\section{Related Work}
\label{sec:related}

There is a large body of literature on bounded model checking that is concerned with encoding temporal logic specifications into propositional logic, see \cite{bmc,bmc2} as starting points.
This line of work is clearly orthogonal to ours.

Moreover, numerous techniques focus on proving \emph{safety} or \emph{satisfiability} of transition systems or CHCs, respectively (e.g., \cite{spacer,GPDR,eldarica,ultimate-chc,freqhorn,synthhorn}).
A comprehensive overview is beyond the scope of this paper.
Instead, we focus on techniques that, like ABMC, aim to prove unsafety by finding long counterexamples.

The most closely related approach is \emph{Acceleration Driven Clause Learning} \cite{adcl,adcl-nt}, a calculus that uses depth-first search and acceleration to find counterexamples.
So one major difference between ABMC and ADCL is that ABMC performs breadth-first search, whereas ADCL performs depth-first search.
Thus, ADCL requires a mechanism for backtracking to avoid getting stuck.
To this end, it relies on a notion of \emph{redundancy}, which is difficult to automate.
Thus, in practice, approximations are used \cite[Sect.~4]{adcl}.
However, even with a complete redundancy check, ADCL might get stuck in a safe part of the search space \cite[Thm.~18]{adcl}.
ABMC does not suffer from such deficits.

Like ADCL, ABMC also tries to avoid redundant work (see \Cref{sec:tuning,sec:blocking}).
However, doing so is crucial for ADCL due to its depth-first strategy, whereas it is a mere optimization for ABMC.

On the other hand, ADCL applies acceleration in a very systematic way, whereas ABMC decides whether to apply acceleration or not based on the model that is found by the underlying SMT solver.
Therefore, ADCL is advantageous for examples with deeply nested loops, where ABMC may require many steps until the SMT solver yields models that allow for accelerating the nested loops one after the other.
Furthermore, ADCL has successfully been adapted for proving non-termination \cite{adcl-nt}, and it is unclear whether a corresponding adaption of ABMC would be competitive.
Thus, both techniques are orthogonal.
See \Cref{sec:experiments} for an experimental comparison of ADCL with ABMC.

Other acceleration-based approaches \cite{fast,flata,journal} can be seen as generalizations of the classical state elimination method for finite automata:
Instead of transforming finite automata to regular expressions, they transform transition systems to formulas that represent the runs of the transition system.
During this transformation, acceleration is the counterpart to the Kleene star in the state elimination method.
Clearly, these approaches differ fundamentally from ours.

In \cite{underapprox15}, under-approximating acceleration techniques are used to enrich the control-flow graph of \pl{C} programs.
Then an external model checker is used to find counterexamples. 
In contrast, ABMC tightly integrates acceleration into BMC, and thus enables an interplay of both techniques:
Acceleration changes the state of the bounded model checker by adding learned transitions to the SMT problem.
Vice versa, the state of the bounded model checker triggers acceleration.
Doing so is impossible if the bounded model checker is used as an external black box.

In \cite{kroening15}, the approach from \cite{underapprox15} is extended by a program transformation that, like our blocking clauses, rules out superfluous traces.
For structured programs, program transformations are quite natural.
However, as we analyze unstructured transition formulas, such a transformation would be
very expensive in our setting.
More precisely, \cite{kroening15} represents programs as CFAs.
To transform them, the edges of the CFA are inspected.
In our setting, the syntactic implicants correspond to these edges.
An important goal of ABMC is to avoid computing them explicitly.
Hence, it is unclear how to apply the approach from \cite{kroening15} in our setting.

Another related approach is described in \cite{iosif12}, where acceleration is integrated into a CEGAR loop in two ways: (1) as preprocessing and (2) to generalize interpolants.
In contrast to (1), we use acceleration ``on the fly''.
In contrast to (2), we do not use abstractions, so our learned transitions can directly be used in counterexamples.
Moreover, \cite{iosif12} only applies acceleration to conjunctive transition formulas, whereas we accelerate conjunctive variants of arbitrary transition formulas.
So in our approach, acceleration techniques are applicable more often, which is particularly useful for finding long counterexamples.

Finally, \emph{transition power abstraction} (TPA) \cite{golem} computes a sequence of over-ap\-prox\-i\-ma\-tions for transition systems where the $n^{th}$ element captures $2^n$ instead of just $n$ steps of the transition relation.
So like ABMC, TPA can help to find long refutations quickly, but in contrast to ABMC, TPA relies on over-approximations.

\section{Experiments and Conclusion}
\label{sec:experiments}

We presented ABMC, which integrates acceleration techniques into bounded model checking.
By enabling BMC to find deep counterexamples, it targets a major limitation of BMC.
However, whether ABMC makes use of transitions that result from acceleration depends on the models found by the underlying SMT solver.
Hence, we introduced \emph{blocking clauses} to enforce the use of accelerated transitions, which also enable ABMC to prove safety in cases where BMC fails.

We implemented ABMC in our tool \tool{LoAT} \cite{loat}.
It uses the SMT solvers \tool{Z3} \cite{z3} and \tool{Yices} \cite{yices}.
Currently, our implementation is restricted to integer arithmetic.
It uses the acceleration technique from \cite{acceleration-calculus} which, in our
experience, is precise in most cases where the values of the variables after executing the
loop can be expressed by polynomials of degree $\leq 2$ (i.e., here we have
${\to_{\accel(\tau)}} = {\to_{\tau}^+}$).
If acceleration yields a non-polynomial formula, then this formula is discarded by our implementation, since \tool{Z3} and \tool{Yices} only support polynomials.
We evaluate our approach on the examples from the category LIA-Lin (linear
CHCs with linear integer arithmetic)\footnote{The restriction
of our approach to linear clauses (with at most one negative literal) is ``inherited''
from BMC.
In contrast, our approach also supports
non-linear arithmetic, but we are not aware of corresponding benchmark collections.}
from the CHC competition~'23 \cite{CHC-COMP}, which
contain problems from numerous applications like verification of
\href{https://github.com/chc-comp/hcai-bench}{\tool{C}},
\href{https://github.com/chc-comp/rust-horn}{\tool{Rust}},
\href{https://github.com/chc-comp/jayhorn-benchmarks}{\tool{Java}}, and
\href{https://github.com/chc-comp/hopv}{higher-order} programs, and
\href{https://github.com/mattulbrich/llreve}{regression verification of \tool{LLVM} programs},
see \cite{chc-comp23} for details. 
By using CHCs as input format, our approach can be used by any CHC-based tool like \tool{Korn} \cite{korn} and \tool{SeaHorn} \cite{seahorn} for \pl{C} and \CPP{} programs, \tool{JayHorn} for \pl{Java} programs \cite{jayhorn}, \tool{HornDroid} for \pl{Android} \cite{horndroid}, \tool{RustHorn} for \pl{Rust} programs \cite{rusthorn}, and \tool{SmartACE} \cite{smartACE} and \tool{SolCMC} \cite{solcmc} for \pl{Solidity}.

We compared several configurations of \tool{LoAT} with the techniques of other leading CHC solvers.
More precisely, we evaluated the following configurations:
\begin{description}
\item[\tool{LoAT}] We used \tool{LoAT}'s implementations of \Cref{alg:bmc} (\tool{LoAT} $\BMC$), \Cref{alg:abmc} (\tool{LoAT} $\ABMC$), \Cref{alg:abmc-blocking} (\tool{LoAT} $\ABMC_\block$), and ADCL (\tool{LoAT ADCL}).
\item[\tool{Z3} \cite{z3}] We used \tool{Z3} 4.13.0, where we evaluated its implementations of the Spacer algorithm (\tool{Spacer} \cite{spacer}) and BMC (\tool{Z3 BMC}).
\item[\tool{Golem} \cite{golem}] We used \tool{Golem} 0.5.0, where we evaluated its
  implementations of
\emph{transition power abstraction}
 (\tool{Golem TPA} \cite{golem}) and BMC (\tool{Golem BMC}).
\item[\tool{Eldarica} \cite{eldarica}] We used \tool{Eldarica} 2.1.0. We tested all five configurations that are used in parallel in its portfolio mode ({\tt -portfolio}), and included the two that found the most counterexamples: CEGAR with acceleration as preprocessing (\tool{Eldarica CEGAR}, {\tt eld -splitClauses:1 -abstract:off -stac}) and symbolic execution (\tool{Eldarica SYM}, {\tt eld -splitClauses:1 -sym}).
\end{description}
Note that all configurations except \tool{Spacer} and \tool{Eldarica CEGAR} are specifically designed for finding counterexamples.
We did not include further techniques for proving safety in our evaluation,
as our focus is on \emph{dis}proving safety.
We ran our experiments on \tool{StarExec} \cite{starexec} with a wallclock timeout of $300$s, a cpu timeout of $1200$s, and a memory limit of $128$GB per example.

\begin{figure}[th!]
  \hspace{1.25em}
  \begin{minipage}{0.38\textwidth}
    \begin{tabular}{|c||c|c||c|c|}
      \hline \multirow{2}{*}{2023} & \multicolumn{2}{c||}{$\unsafe$} & \multicolumn{2}{c|}{$\safe$} \\
      \hhline{~----} & \checkmark & {\bf !} & \checkmark & {\bf !} \\
      \hline\hline \tool{LoAT} $\ABMC$ & 73 & -- & 31 & -- \\
      \hline \tool{LoAT} $\ABMC_\block$ & 72 & 0 & 75 & 11 \\
      \hline \tool{Golem TPA} & 64 & 0 & 83 & 5 \\
      \hline \tool{LoAT} $\BMC$ & 60 & 0 & 36 & 0 \\
      \hline \tool{Z3 BMC} & 57 & -- & 21 & -- \\
      \hline \tool{LoAT ADCL} & 56 & 1 & 0 & -- \\
      \hline \tool{Golem BMC} & 55 & -- & 20 & -- \\
      \hline \tool{Spacer} & 51 & 4 & 151 & 53 \\
      \hline \tool{Eldarica CEGAR} & 46 & 1 & 107 & 13 \\
      \hline \tool{Eldarica SYM} & 46 & 1 & 68 & 15 \\
      \hline
    \end{tabular}
  \end{minipage}
  \begin{minipage}{0.615\textwidth}
    \begin{tikzpicture}[scale=0.8]
      \begin{axis}[
        legend pos=south east,
        ylabel=$\unsafe$ proofs,
        xticklabel={$\pgfmathprintnumber{\tick}$s},
        ymin=31]
        \addplot[color=violet,densely dashed,thick] table[col sep=comma,header=false,x index=0,y index=1] {loat_non_blocking_abmc_23.csv};
        \addlegendentry{\tool{LoAT} $\ABMC$}
        \addplot[color=black,solid,thick] table[col sep=comma,header=false,x index=0,y index=1] {loat_abmc_23.csv};
        \addlegendentry{\tool{LoAT} $\ABMC_\block$}
        \addplot[color=blue,loosely dashed,thick] table[col sep=comma,header=false,x index=0,y index=1] {golem_tpa_23.csv};
        \addlegendentry{\tool{Golem TPA}}
        \addplot[color=red,dotted,thick] table[col sep=comma,header=false,x index=0,y index=1] {loat_bmc_23.csv};
        \addlegendentry{\tool{LoAT} $\BMC$}
        \addplot[color=lightgray,solid,thick] table[col sep=comma,header=false,x index=0,y index=1] {z3_bmc_23.csv};
        \addlegendentry{\tool{Z3 BMC}}
        \addplot[color=OliveGreen,loosely dotted,thick] table[col sep=comma,header=false,x index=0,y index=1] {loat_adcl_23.csv};
        \addlegendentry{\tool{LoAT ADCL}}
      \end{axis}
    \end{tikzpicture}
  \end{minipage}
\end{figure}

The results can be seen in the table above.
The columns with {\bf !} show the num\-ber of unique proofs, i.e., the number of examples that could only be solved by the corresponding configuration.
Such a comparison only makes sense if just one implementation of each algorithm is
considered. For instance, \tool{LoAT}'s, \tool{Z3}'s, and \tool{Golem}'s implementations
of the BMC algorithm work well on the same class of examples, so that none of them finds unique proofs if all of them are taken into account.
Thus, for {\bf !} we disregarded \tool{LoAT} $\ABMC$, \tool{Z3 BMC}, and \tool{Golem BMC}.

The table shows that our implementation of ABMC is very powerful for proving unsafety.
In particular, it shows a significant improvement over \tool{LoAT BMC}, which is implemented very similarly, but does not make use of acceleration.

Note that all unsafe instances that can be solved by ABMC can also be solved by other configurations.
This is not surprising, as \tool{LoAT ADCL} is also based on acceleration techniques.
Hence, ABMC combines the strengths of ADCL and BMC, and conversely, unsafe examples that can be solved with ABMC can usually also be solved by one of these techniques.
So for unsafe instances, the main contribution of ABMC is to have \emph{one} technique
that performs well both on instances with shallow counterexamples (which can be solved by BMC) as well as instances with deep counterexamples only (which can often be solved by ADCL).

On the instance that can only be solved by ADCL, our (A)BMC implementation spends most of the time with applying substitutions, which clearly shows potential for further optimizations.
Due to ADCL's depth-first strategy, it produces smaller formulas, so that applying substitutions is cheaper.

Regarding safe examples, the table shows that our implementation of ABMC is not
competitive with state-of-the-art techniques.\footnote{\tool{LoAT ABMC}
finds fewer safety proofs than \tool{LoAT BMC}  since acceleration
sometimes yields transitions with non-linear arithmetic that make the SMT
problem harder.}
However, it finds several unique proofs.
This is remarkable, as \tool{LoAT} is not at all fine-tuned for proving safety.
For example, we expect that \tool{LoAT}'s results on safe instances can easily be improved by integrating over-approximating acceleration techniques.
While such a variant of ABMC could not prove unsafety, it would presumably be much more powerful for proving safety.
We leave that to future work.

The plot on the previous page shows how many unsafety proofs were found within 300~s, where we only
include the six best configurations for readability. It shows that ABMC is highly competitive on unsafe instances, not only in
\begin{wrapfigure}[11]{r}{0.42\textwidth}

\vspace*{-.1cm}  

\hspace{-9pt}\begin{tikzpicture}[scale=0.8]
\begin{axis}[
    enlargelimits=false,
    xlabel={$\ABMC_\block$},
    ylabel={$\BMC$},
    height=6cm,
    xmax=16000,
    ymax=16000,
    log origin=infty,
    xmode=log,
    ymode=log
]
\addplot+[
    only marks,
    mark=x,
    mark size=2pt,
    color=black
    ]
table
{scatter.dat};
\draw
(axis cs:1,1) --
(axis cs:16000,16000);
\end{axis}
\end{tikzpicture}
\end{wrapfigure}
terms of solved examples, but also in terms of runtime.
The plot on the right compares the length of the counterexamples found by \tool{LoAT} $\ABMC_\block$ and $\BMC$ to show the impact of acceleration.
Here, only examples where both techniques disprove safety are considered, and the counterexamples found by $\ABMC_\block$ may contain accelerated transitions.
There are no points below the diagonal, i.e., the counterexamples found by $\ABMC_\block$ are at most as long as those found by $\BMC$.
The points above the diagonal indicate that the counterexamples found by $\ABMC_\block$
are sometimes shorter by orders of magnitude (note that the axes are log-scaled).

Our results also show that blocking clauses have no significant impact on ABMC's performance on unsafe instances, neither regarding the number of solved examples, nor regarding the runtime.
In fact, $\ABMC_\block$ solved one instance less
than $\ABMC$ (which can, however, also be solved by $\ABMC_\block$ with a larger timeout).
On the other hand, blocking clauses are clearly useful for proving safety, where they even allow \tool{LoAT} to find several unique proofs.

In future work, we plan to support other theories like reals, bitvectors, and arrays, and
we will investigate an extension to non-linear CHCs.
Our implementation is open-source and available on Github.
For the sources, a pre-compiled binary, and more information on our evaluation, we refer
to \cite{website}.

\subsubsection*{Data Availability Statement}

An artifact containing \tool{LoAT} which allows to replicate our experiments is
available at \cite{artifact}.

\bibliographystyle{splncs04}
\paper{
  \bibliography{refs,crossrefs,strings}

\begin{thebibliography}{10}
\providecommand{\url}[1]{\texttt{#1}}
\providecommand{\urlprefix}{URL }
\providecommand{\doi}[1]{https://doi.org/#1}

\bibitem{artifact}
Artifact for ``{I}ntegrating {L}oop {A}cceleration into {B}ounded {M}odel
  {C}hecking'' (2024). \doi{10.5281/zenodo.11954015}

\bibitem{website}
Evaluation of ``{I}ntegrating {L}oop {A}cceleration into {B}ounded {M}odel
  {C}hecking'' (2024), \url{https://loat-developers.github.io/abmc-eval/}

\bibitem{solcmc}
Alt, L., Blicha, M., Hyv{\"{a}}rinen, A.E.J., Sharygina, N.: \tool{SolCMC}:
  \pl{Solidity} compiler's model checker. In: CAV~'22. pp. 325--338. LNCS 13371
  (2022). \doi{10.1007/978-3-031-13185-1\_16}

\bibitem{fast}
Bardin, S., Finkel, A., Leroux, J., Petrucci, L.: {\textsf{FAST}:}
  {A}cceleration from theory to practice. Int. J. Softw. Tools Technol. Transf.
   \textbf{10}(5),  401--424 (2008). \doi{10.1007/s10009-008-0064-3}

\bibitem{bmc2}
Biere, A., Cimatti, A., Clarke, E.M., Strichman, O., Zhu, Y.: Bounded model
  checking. Advances in Computers  \textbf{58},  117--148 (2003).
  \doi{10.1016/S0065-2458(03)58003-2}

\bibitem{bmc}
Biere\noopsort{1}, A.: Bounded model checking. In: Handbook of Satisfiability -
  Second Edition, pp. 739--764. Frontiers in Artificial Intelligence and
  Applications 336, {IOS} Press (2021). \doi{10.3233/FAIA201002}

\bibitem{golem}
Blicha, M., Fedyukovich, G., Hyv{\"{a}}rinen, A.E.J., Sharygina, N.: Transition
  power abstractions for deep counterexample detection. In: TACAS~'22. pp.
  524--542. LNCS 13243 (2022). \doi{10.1007/978-3-030-99524-9\_29}

\bibitem{bozga09a}
Bozga, M., G{\^{\i}}rlea, C., Iosif, R.: Iterating octagons. In: TACAS~'09. pp.
  337--351. LNCS 5505 (2009). \doi{10.1007/978-3-642-00768-2\_29}

\bibitem{flata}
Bozga\noopsort{1}, M., Iosif, R., Kone{\v{c}}n{\'{y}}, F.: Relational analysis
  of integer programs. Tech. Rep. TR-2012-10, VERIMAG (2012),
  \url{https://www-verimag.imag.fr/TR/TR-2012-10.pdf}

\bibitem{horndroid}
Calzavara, S., Grishchenko, I., Maffei, M.: \tool{HornDroid}: Practical and
  sound static analysis of \tool{Android} applications by {SMT} solving. In:
  EuroS{\&}P~'16. pp. 47--62. {IEEE} (2016). \doi{10.1109/EuroSP.2016.16}

\bibitem{CHC-COMP}
{CHC Competition}, \url{https://chc-comp.github.io}

\bibitem{chc-comp23}
{De Angelis}, E., {Govind V K}, H.: {CHC-COMP} 2023: Competition report (2023),
  \url{https://chc-comp.github.io/2023/CHC_COMP_2023_Competition_Report.pdf}

\bibitem{ultimate-chc}
Dietsch, D., Heizmann, M., Hoenicke, J., Nutz, A., Podelski, A.:
  \textsf{Ultimate {TreeAutomizer}} {(CHC-COMP} tool description). In:
  HCVS/PERR@ETAPS '19. pp. 42--47. EPTCS 296 (2019). \doi{10.4204/EPTCS.296.7}

\bibitem{yices}
Dutertre, B.: \tool{Yices} 2.2. In: CAV~'14. pp. 737--744. LNCS 8559 (2014).
  \doi{10.1007/978-3-319-08867-9\_49}

\bibitem{enderton}
Enderton, H.B.: A Mathematical Introduction to Logic. Academic Press (1972)

\bibitem{korn}
Ernst, G.: Loop verification with invariants and contracts. In: VMCAI~'22. pp.
  69--92. LNCS 13182 (2022). \doi{10.1007/978-3-030-94583-1\_4}

\bibitem{freqhorn}
Fedyukovich, G., Prabhu, S., Madhukar, K., Gupta, A.: Solving constrained
  {Horn} clauses using syntax and data. In: FMCAD~'18. pp.~1--9 (2018).
  \doi{10.23919/FMCAD.2018.8603011}

\bibitem{acceleration-calculus}
Frohn, F.: A calculus for modular loop acceleration. In: TACAS~'20. pp. 58--76.
  LNCS 12078 (2020). \doi{10.1007/978-3-030-45190-5\_4}

\bibitem{journal}
Frohn, F., Naaf, M., Brockschmidt, M., Giesl, J.: Inferring lower runtime
  bounds for integer programs. {ACM} Trans. Program. Lang. Syst.
  \textbf{42}(3),  13:1--13:50 (2020). \doi{10.1145/3410331}

\bibitem{loat}
Frohn\noopsort{4}, F., Giesl, J.: Proving non-termination and lower runtime
  bounds with \tool{LoAT} (system description). In: IJCAR~'22. pp. 712--722.
  LNCS 13385 (2022). \doi{10.1007/978-3-031-10769-6\_41}

\bibitem{adcl-nt}
Frohn\noopsort{5}, F., Giesl, J.: Proving non-termination by {Acceleration
  Driven Clause Learning}. In: CADE~'23. pp. 220--233. LNCS 14132 (2023).
  \doi{10.1007/978-3-031-38499-8\_13}

\bibitem{adcl}
Frohn\noopsort{6}, F., Giesl, J.: {ADCL}: {A}cceleration {D}riven {C}lause
  {L}earning for constrained {H}orn clauses. In: SAS~'23. pp. 259--285. LNCS
  14284 (2023). \doi{10.1007/978-3-031-44245-2\_13}

\bibitem{report}
Frohn\noopsort{6}, F., Giesl, J.: Integrating loop acceleration into bounded
  model checking. CoRR  \textbf{abs/2401.09973} (2024).
  \doi{10.48550/arXiv.2401.09973}

\bibitem{seahorn}
Gurfinkel, A., Kahsai, T., Komuravelli, A., Navas, J.A.: The \tool{SeaHorn}
  verification framework. In: CAV~'15. pp. 343--361. LNCS 9206 (2015).
  \doi{10.1007/978-3-319-21690-4\_20}

\bibitem{GPDR}
Hoder, K., Bj{\o}rner, N.S.: Generalized property directed reachability. In:
  {SAT} '12. pp. 157--171. LNCS 7317 (2012).
  \doi{10.1007/978-3-642-31612-8\_13}

\bibitem{iosif12}
Hojjat, H., Iosif, R., Kone\v{c}n\'{y}, F., Kuncak, V., R{\"{u}}mmer, P.:
  Accelerating interpolants. In: ATVA~'12. pp. 187--202. LNCS 7561 (2012).
  \doi{10.1007/978-3-642-33386-6\_16}

\bibitem{eldarica}
Hojjat, H., R{\"{u}}mmer, P.: The \tool{Eldarica} {Horn} solver. In: FMCAD~'18.
  pp.~1--7 (2018). \doi{10.23919/FMCAD.2018.8603013}

\bibitem{jayhorn}
Kahsai, T., R{\"{u}}mmer, P., Sanchez, H., Sch{\"{a}}f, M.: \tool{JayHorn}: {A}
  framework for verifying \pl{Java} programs. In: CAV~'16. pp. 352--358. LNCS
  9779 (2016). \doi{10.1007/978-3-319-41528-4\_19}

\bibitem{spacer}
Komuravelli, A., Gurfinkel, A., Chaki, S.: {SMT}-based model checking for
  recursive programs. Formal Methods Syst. Des.  \textbf{48}(3),  175--205
  (2016). \doi{10.1007/s10703-016-0249-4}

\bibitem{underapprox15}
Kroening, D., Lewis, M., Weissenbacher, G.: Under{-}approximating loops in
  \pl{C} programs for fast counterexample detection. Formal Methods Syst. Des.
  \textbf{47}(1),  75--92 (2015). \doi{10.1007/s10703-015-0228-1}

\bibitem{kroening15}
Kroening\noopsort{5}, D., Lewis, M., Weissenbacher, G.: Proving safety with
  trace automata and bounded model checking. In: {FM}~'15. pp. 325--341. LNCS
  9109 (2015). \doi{10.1007/978-3-319-19249-9\_21}

\bibitem{rusthorn}
Matsushita, Y., Tsukada, T., Kobayashi, N.: \tool{RustHorn}: {CHC}-based
  verification for \pl{Rust} programs. {ACM} Trans. Program. Lang. Syst.
  \textbf{43}(4),  15:1--15:54 (2021). \doi{10.1145/3462205}

\bibitem{z3}
\noopsort{Moura}{de Moura}, L., Bj{\o}rner, N.: \tool{Z3}: An efficient {SMT}
  solver. In: TACAS\ '08. pp. 337--340. LNCS 4963 (2008).
  \doi{10.1007/978-3-540-78800-3\_24}

\bibitem{starexec}
Stump, A., Sutcliffe, G., Tinelli, C.: \tool{StarExec}: {A} cross-community
  infra\-structure for logic solving. In: IJCAR~'14. pp. 367--373. LNCS 8562
  (2014). \doi{10.1007/978-3-319-08587-6\_28}

\bibitem{smartACE}
Wesley, S., Christakis, M., Navas, J.A., Trefler, R.J., W{\"{u}}stholz, V.,
  Gurfinkel, A.: Verifying \pl{Solidity} smart contracts via communication
  abstraction in \tool{SmartACE}. In: VMCAI~'22. pp. 425--449. LNCS 13182
  (2022). \doi{10.1007/978-3-030-94583-1\_21}

\bibitem{synthhorn}
Zhu, H., Magill, S., Jagannathan, S.: A data-driven {CHC} solver. In: {PLDI}
  '18. pp. 707--721 (2018). \doi{10.1145/3192366.3192416}

\end{thebibliography}
}
\report{
  \providecommand{\noopsort}[1]{}

  \clearpage 
  \begin{appendix}
\appendixproofsection{Appendix}\label{sec:MissingProofs}
\appendixproof*{thm:properties}
\end{appendix}
}

\end{document}